%% file: main.tex
\documentclass[3p,times,review]{elsarticle}
\newif\ifreview
\reviewfalse

\input{main_common}

\journal{Chaos, Solitons \& Fractals}
\begin{document}

\begin{frontmatter}

\title{
    Simulating neuronal dynamics in fractional adaptive exponential
    integrate-and-fire models}

\author[icamaddress]{Alexandru Fikl\corref{correspondingauthor}}
\cortext[correspondingauthor]{Corresponding Author}
\ead{alexandru.fikl@e-uvt.ro}

\author[uvtaddress]{Aman Jhinga}
\author[uvtaddress]{Eva Kaslik}
\author[skbaddress]{Argha Mondal}

\address[icamaddress]{
    Institute for Advanced Environmental Research,
    West University of Timișoara, 300086, Romania}
\address[uvtaddress]{
    Department of Mathematics and Computer Science,
    West University of Timișoara, Timișoara 300223, Romania}
\address[skbaddress]{
    Department of Mathematics,
    Sidho-Kanho-Birsha University, Purulia-723104, WB, India}

\begin{abstract}
We introduce an efficient discretization of a novel fractional-order adaptive
exponential (\FrAdEx{}) integrate-and-fire model, which is used to study the
fractional-order dynamics of neuronal activities. The discretization is based on
extension of L1-type methods that can accurately handle the exponential growth
and the spiking mechanism of the model. This new method is implicit
and uses adaptive time stepping to robustly handle the stiff system that arises
due to the exponential term. The implicit nonlinear system can be solved exactly,
without the need for iterative methods, making the scheme efficient while
maintaining accuracy. We present a complete error model for the numerical scheme
that can be extended to other integrate-and-fire models with minor changes. To
show the feasibility of our approach, the numerical method has been rigorously
validated and used to investigate several different spiking oscillations of the
model. We observed that the fractional-order model is capable of predicting
biophysical activities, which are interpreted through phase diagrams describing
the transition from one firing type to another. This simple model shows significant
promise, as it has sufficient expressive dynamics to reproduce several features
qualitatively from a biophysical dynamical perspective.
\end{abstract}


\end{frontmatter}

\section{Introduction}
\label{sc:introduction}

The adaptive exponential (AdEx) integrate-and-fire model \cite{brette2005adaptive}
holds an important position in the field of computational neuroscience, serving
as a bridge between simplified neuron models and highly complex real-world neuronal
dynamics \cite{naud2008firing,touboul2008dynamics}. Originating
from classical integrate-and-fire models \cite{izhikevich2007dynamical,lapicque1907recherches},
the AdEx variant incorporates several mechanisms to better replicate the
spiking patterns of real biological neurons \cite{gorski2021conductance}. To
further enhance the fidelity of the AdEx models, authors have recently looked at
introducing non-standard operators with desirable properties, such as fractional
derivatives \cite{teka2014neuronal}.

Its precursor, the traditional leaky integrate-and-fire model \cite{izhikevich2007dynamical},
has been widely employed as a simple representation of spiking neurons, but has
showcased multiple limitations. To address these limitations, researchers
have proposed generalisations in three distinct directions. Firstly, the model
was extended to incorporate non-linear components, such as
quadratic \cite{izhikevich2003simple} or exponential terms \cite{gerstner2014neuronal}.
This modification allows replacing the strict voltage threshold with a more
realistic smooth spike initiation zone, which better resembles the behaviour of
real neurons. Secondly, an additional variable has been introduced into the model
\cite{lansky2008review}, which accounts for sub-threshold resonances or adaptation
mechanisms. Sub-threshold resonances capture how a neuron's response to input
varies with input frequency and adaptation mechanisms reflect changes in a neuron's
firing rate or behaviour in response to prolonged stimulation. A third direction
involves changing the stimulation paradigm from current injection to conductance
injection with the aim of bringing integrate-and-fire models closer to the
conditions experienced by cortical neurons in vivo
\cite{fourcaud2003spike,gerstner2014neuronal}. Conductance injection takes into
consideration the dynamic changes in synaptic conductance that occur in real
neuronal networks, enhancing the physiological relevance of the model. These
extensions have been incorporated into the AdEx model, making it more suitable
for simulating the behaviour of spiking neurons in diverse biological systems.

Detailed conductance-based models, such as Hodgkin-Huxley-type models, have
also been employed to investigate electrophysiological phenomena
\cite{izhikevich2004model,prinz2003alternative}. However, these models pose
challenges for direct mathematical analysis, numerical implementation, and
the replication of experimental spike patterns. Therefore, studies of integrate-and-fire
models, such as the AdEx model, have been performed and have shown promising
results compared to more complex neuron models
\cite{brette2005adaptive,dayan2005theoretical}. In particular, they demonstrate
their effectiveness in reproducing the behaviour of layer-V neocortical pyramidal
neurons under random stimuli \cite{clopath2007predicting}.

Very recently, an extension of the AdEx model has been proposed in \cite{souza2024adaptive},
which involves replacing the standard integer-order derivative with a so-called
(local) fractal derivative. The study suggests that the fractal derivative orders
significantly influence inter-spike intervals and mean firing frequencies,
potentially offering a more accurate representation of neuronal activity. We
investigate here a similar extension involving the fractional (Caputo) derivative,
which lends itself to a more rigorous mathematical analysis. The relevance of
time-fractional calculus has grown immensely in various interdisciplinary fields
over the past few decades due to its ability to incorporate memory effects and hereditary
characteristics inherent in many physical systems
\cite{podlubny1998fractional,samko1993fractional}. Fractional systems have
also been found to exhibit richer dynamics than some of their integer-order
counterparts\cite{Li2004, Yu2009}. As such, fractional calculus has been effectively
employed in modelling a wide spectrum of phenomena, from anomalous transport in
porous media \cite{Hilfer,Metzler} to complex behaviours observed in biological
systems \cite{ionescu2017role,Langlands2011,magin2010fractional}.

Incorporating fractional derivatives into integrate-and-fire neuronal models
introduces a new mathematical tool directed at capturing the nonlocal, history-dependent
behaviour of neuronal dynamics. Traditional models, based on integer-order
derivatives, fall short in modelling memory effects and continuous integration
of past neuronal states. This limitation is linked to the inability of
classical models to simulate long-term dependencies or anomalous propagation
behaviour \cite{Hilfer,Metzler} that are characteristics of certain neuronal systems.
To address these challenges, a fractional order leaky integrate-and-fire model
was introduced \cite{teka2014neuronal}, recognising that neuronal voltage trajectories
often exhibit dynamics across multiple time scales
\cite{gilboa2005history,weinberg2017history}. These dynamics, indicative of
complex intracellular interactions, often follow power-law behaviours that may be
efficiently modelled by fractional differential equations \cite{Deng2022}.

In this paper, we describe a fractional-order adaptive exponential (\FrAdEx{})
model, designed to investigate the fractional-order dynamics
that govern neuronal activity. As shown, the \FrAdEx{} model \eqref{eq:model:fde}
with the reset condition \eqref{eq:model:reset} is an impulsive system of
(Caputo-type) fractional differential equations with state-dependent impulses.
Impulsive differential equations offer a valuable mathematical framework for
investigating evolutionary processes characterised by abrupt and sudden state
changes. They can be categorised into two main classes based on
their impulsive events: fixed-time impulsive systems and state-dependent impulsive
systems. Although extensive research efforts have been dedicated to the study of
fixed-time impulsive systems \cite{wang2016survey}, it is crucial to acknowledge
that in the real world many systems do not experience impulses at predetermined
fixed intervals,  but rather in response to specific states
\cite{Lakshmikantham1994,yang2018state}. To the best of our knowledge, there
are currently no published results in the literature regarding the theoretical and numerical analysis of fractional-order
impulsive systems with state-dependent impulses.

Consequently, the objective of this paper is to propose and investigate numerical
methods tailored for the \FrAdEx{} model. Specifically, we describe an implicit
adaptive L1-type scheme engineered to address the three main challenges in
simulating the \FrAdEx{} model. Firstly, the (Caputo) fractional derivative is
defined as an integral operator with a weakly singular kernel which we discretise
using the standard linear interpolation of the L1 method \cite{li2015numerical}.
This allows a straightforward handling of the singularity and results in an implicit
numerical scheme. An implicit method is desired for the \FrAdEx{} model due to the
exponential growth in the spike initiation zone, which results in a generally stiff
system. Secondly, using an implicit method is not sufficient to accurately represent
the exponential growth, so the L1 method is augmented with an adaptive time-stepping
scheme based on \cite{Jannelli2020}. The non-uniform time step allows a fine
control over the error in the exponential region and results in accurate
estimates of the spike times. The third issue is common to state-dependent
impulsive systems and requires an accurate representation of the reset condition
as the neuron produces spikes. Discretizations of integer-order integrate-and-fire
models have been investigated in \cite{Shelley2001}, where a high-order
approximation of the spike times and reset is provided. However, such extensions
are not possible for the \FrAdEx{} model due to the memory properties of the
fractional derivative. Therefore, we present an estimate of the spike times
based on the Lambert W function that does not require iteration. High-order methods
for fractional differential equations have also been developed
\cite{li2015numerical,Li2017,Yang2023}, but much less attention has been shown
to the fully non-uniform case required here. As such, significant effort will be
required to extend the current discretization to higher-order. The first-order
method presented here features enhanced stability properties and a complete error model
that can be extended to other integrate-and-fire-type models with minor
modifications.

The structure of the paper is as follows. In \Cref{sc:model}, we will introduce
the \FrAdEx{} model and describe its parameters and general construction. In
\Cref{sc:fractional}, we will briefly define the necessary fractional calculus
concepts and notation to interpret the \FrAdEx{} model, which is used
in \Cref{sc:dim} to present a nondimensional version of the model that is used
in simulations. Then, \Cref{sc:methods} describes the
L1-type scheme used to discretise the model and the specific handling of the
state-dependent impulses. The stability and convergence of the method is analysed
in \Cref{sc:error}. We then verify and validate the method on multiple
integrate-and-fire models and show its performance properties in \Cref{sc:results}.
In \Cref{sc:behaviour}, we show that the newly introduced \FrAdEx{} model can
numerically reproduce known phenomenological responses of neurons for several
parameter ranges. Finally, we discuss the conclusions in \Cref{sc:conclusions}.

\section{Model}
\label{sc:model}

The fractional-order model we consider here is a straightforward extension of
the standard integer-order \textrm{AdEx} model from \cite{naud2008firing}. This
extension is common in capacitance-based models and relies on the generalisation of
Curie's empirical law described in \cite{Westerlund1994capacitor}. It is given by
\begin{equation} \label{eq:model:fde}
\begin{aligned}
C \od{^{\alpha_1} V}{t^{\alpha_1}} & =
    I
    - g_L (V - E_L)
    + g_L \Delta_T \exp \left(\frac{V - V_T}{\Delta_T}\right)
    - w, \\
\tau_w \od{^{\alpha_2} w}{t^{\alpha_2}} & =
    a (V - E_L) - w,
\end{aligned}
\end{equation}
under suitable initial conditions with fractional orders $0 < \alpha_i < 1$,
not necessarily equal (for details, see \Cref{sc:fractional}). The \FrAdEx{} model
describes the evolution of the membrane voltage potential $V(t)$, driven by an
external current $I(t)$, and an adaptation variable $w(t)$. The model is constructed
such that when the voltage potential increases beyond the parameter value $V_T$,
the exponential term in the first equation activates a positive feedback for
spike generation. In integrate-and-fire models, this exponential increase in the
potential is truncated at a threshold value $V_{\text{peak}}$, which is used to
denote the generation of an action potential. At this point, the voltage is
reset to a fixed value $V_r$ and the adaptation variable $w$ is updated by a
fixed offset $b$. In physical systems, the neuron exhibits a spike if $V$ grows
rapidly. However, in the \FrAdEx{} model, a spike is generated only when it reaches
the threshold value $V_{\text{peak}}$. This may shift spike times by a very
small time duration (milliseconds) compared to
experiments \cite{clopath2007predicting,naud2008firing}. The model \eqref{eq:model:fde}
uses the following reset condition to control the upswing and downswing of the
membrane potential
\begin{equation} \label{eq:model:reset}
\text{if $V > V_{\text{peak}}$ then }
\begin{cases}
V \gets V_r, \\
w \gets w + b.
\end{cases}
\end{equation}

The parameters describing the model are divided into so-called \emph{scaling}
parameters and \emph{bifurcation} parameters. The scaling parameters are responsible
for scaling the time axis and scaling (or offsetting) the system variables
\cite{naud2008firing}. The five scaling parameters are the total capacitance
$C$, the total leak conductance $g_L$, the effective resting potential $E_L$,
the threshold slope factor $\Delta_T$, and the effective threshold potential $V_T$.
The bifurcation parameters are directly proportional to the time constant $\tau_w$,
the conductance $a$, the spike triggered adaptation $b$ and the reset potential $V_r$.
These parameters are responsible for changes in the qualitative neuronal activities
(such as multiple firings). For example, the parameter $a$ controls the sensitivity
of the adaptation current to the membrane voltage.

\section{Fractional Calculus Background}
\label{sc:fractional}

We give here a short introduction to the notation used throughout the
remaining sections and motivate the choice of fractional derivative operators.
For simplicity, we denote the fractional model \eqref{eq:model:fde} of order
$\vect{\alpha} \in (0, 1)^2$ in vector form as
\[
\od{^{\vect{\alpha}} \vect{y}}{t^{\vect{\alpha}}} = \vect{f}(t, \vect{y}),
\]
where $\vect{y} \triangleq (V, w)$ and $\vect{f}: \mathbb{R}_+ \times \mathbb{R}^2
\to \mathbb{R}^2$ is the right-hand side. We also denote by $\tau_m$, for
$m \in \{0, \dots, M\}$, the $M + 1$ spike times that occur in the system, with
the convention that $\tau_0 = 0$ and $\tau_N = T$ denote the domain limits. The
interior spikes $\tau_m(\vect{y})$, for $m \in \{1, \dots, M - 1\}$,
are state-dependent and are not known a priori. We also use the notation
$\mathrm{AC}([0, T]; \mathbb{R}^n)$ to denote the space of absolutely continuous
functions and $\mathrm{PAC}([0, T]; \mathbb{R}^n)$ to denote the space
of piecewise absolutely continuous functions on $[0, T]$ with values in
$\mathbb{R}^n$. The standard notation $\mathrm{PAC}^k([0, T], \mathbb{R}^n)$
is used to denote functions that are piecewise differentiable $k$ times
and each derivative is absolutely continuous. For impulsive systems,
the solutions are found in the $\mathrm{PAC}^k([0, T]; \mathbb{R}^n)$ space due
to the reset condition. To handle the vector system, we make use of standard
multi-index notation, where operations are applied component-wise, unless
indicated.

The choice of fractional derivative in \eqref{eq:model:fde} is also a choice that
must be made in the modelling process. In this work, we focus on the
Caputo derivative, which is classically defined for $\vect{y} \in
\mathrm{AC}([0, T]; \mathbb{R}^2)$
by \cite{kilbas2006theory}
\[
{}^C D^{\vect{\alpha}}_{0^+}[\vect{y}](t) \triangleq
    \frac{1}{\Gamma(1 - \vect{\alpha})} \int_0^t
    \frac{\vect{y}'(s)}{(t - s)^{\vect{\alpha}}} \dx[s],
\]
for $t > 0$. The Caputo derivative is often used in physical systems because it
has several attractive qualities, such as local initial conditions and the
fact that the derivative of a constant vanishes \cite{kilbas2006theory}.
This definition can be extended to the case of $\vect{y} \in \mathrm{PAC}([0, T];
\mathbb{R}^2)$, as shown in \cite{Feckan2012}. We have that
\begin{equation} \label{eq:fractional:caputo}
\CaputoD[\vect{\alpha}]{\vect{y}}{0^+}(t) \triangleq
\frac{1}{\Gamma(1 - \vect{\alpha})} \left[
    \sum_{j = 0}^{m - 1}
    \int_{\tau_j}^{\tau_{j + 1}} \frac{\vect{y}'(s)}{(t - s)^{\vect{\alpha}}} \dx[s]
    + \int_{\tau_m}^t \frac{\vect{y}'(s)}{(t - s)^{\vect{\alpha}}} \dx[s]
    \right],
\end{equation}
for $t \in (\tau_m, \tau_{m + 1}]$. According to \cite[Lemma 2.7]{Feckan2012},
we can write the impulsive system in Volterra integral equation form as
\begin{equation} \label{eq:fractional:caputo_volterra}
\vect{y}(t) = \vect{y}_0 + \sum_{j = 1}^m [\vect{y}(\tau_j^+) - \vect{y}(\tau_j^-)] +
    \frac{1}{\Gamma(\vect{\alpha})}
    \int_0^t \frac{\vect{f}(s, \vect{y}(s))}{(t - s)^{1 - \vect{\alpha}}} \dx[s],
\end{equation}
where the one-sided limits are defined as
\[
y(\tau_m^\pm) \triangleq \lim_{\epsilon \to 0^+} y(\tau_m \pm \epsilon).
\]

In practice, both formulations can be used in the construction of a
numerical method. However, for the \FrAdEx{} model, we prefer discretising the
Caputo derivative directly, as it avoids the summation over the exponential
right-hand side and gives rise to implicit methods, which allow for larger
time steps.

\begin{remark}
For the \FrAdEx{} model \eqref{eq:model:reset}, we have a constant jump at each
spike time, i.e., $\vect{y}(\tau_m^+) - \vect{y}(\tau_m^-) = \Delta \vect{y}
\triangleq (V_r - V_{\text{peak}}, b)$. Therefore, we can write
\[
\vect{y}(t) = \vect{y}_0 + m \Delta \vect{y} +
    \frac{1}{\Gamma(\vect{\alpha})}
    \int_0^t \frac{\vect{f}(s, \vect{y}(s))}{(t - s)^{1 - \vect{\alpha}}} \dx[s].
\]
\end{remark}

\begin{remark}
The Riemann--Liouville derivative can also be used in the definition
of \eqref{eq:model:fde}. In the piecewise-continuous case, we have
that \cite[Definition 2.4]{Liu2019}
\[
\RiemannD[\vect{\alpha}]{\vect{y}}{0+}(t) =
\CaputoD[\vect{\alpha}]{\vect{y}}{0+}(t)
+ \frac{1}{\Gamma(1 - \vect{\alpha})} \frac{\vect{y}(\tau_0^+)}{t^{\vect{\alpha}}}
+ \frac{1}{\Gamma(1 - \vect{\alpha})} \sum_{j = 1}^m
    \frac{\vect{y}(\tau_j^+) - \vect{y}(\tau_j^-)}{(t - \tau_j)^{\vect{\alpha}}},
\]
for $t \in (\tau_m, \tau_{m + 1}]$. However, while the Riemann--Liouville derivative
requires less regularity in the solutions, it also requires nonlocal initial
conditions \cite{kilbas2006theory}.
\end{remark}

\section{Non-Dimensional Model}
\label{sc:dim}

The model \eqref{eq:model:fde} can be non-dimensionalized to reduce the parameter
space to only $4$ parameters in the equations and the $3$ parameters in the
reset condition. We provide here a non-dimensionalization that we use going forward
to simplify the description and analysis of the numerical methods. We follow the
suggestion from \cite{naud2008firing} and define the non-dimensional variables
\[
\varref{t} \triangleq \sqrt[\alpha_1]{\frac{g_L}{C}} t, \quad
\varref{V} \triangleq \frac{V - V_T}{\Delta_T}
\quad \text{and} \quad
\varref{w} \triangleq \frac{w}{\Delta_T g_L}.
\]

This choice of non-dimensionalization is specific to the fractional case, where
the Caputo derivative itself has units of $\mathrm{s}^{-\alpha}$. Replacing these
relations into our original model \eqref{eq:model:fde}, we obtain
\begin{equation} \label{eq:dim:model}
\left\{
\begin{aligned}
\od{^{\alpha_1} \varref{V}}{\varref{t}^{\alpha_1}} & =
\varref{I} - (\varref{V} - \varref{E}_L) + \exp(\varref{V}) - \varref{w}, \\
\varref{\tau}_w \od{^{\alpha_2} \varref{w}}{\varref{t}^{\alpha_2}} & =
    \varref{a} (\varref{V} - \varref{E}_L) - \varref{w},
\end{aligned}
\right.
\end{equation}
where
\[
\varref{I} \triangleq \frac{I}{\Delta_T g_L},
\quad
\varref{E}_L \triangleq \frac{E_L - V_T}{\Delta_T},
\quad
\varref{\tau}_w \triangleq \left(\frac{g_L}{C}\right)^{\frac{\alpha_2}{\alpha_1}} \tau_w
\quad \text{and} \quad
\varref{a} \triangleq \frac{a}{g_L}.
\]

We also obtain the non-dimensional reset condition
\begin{equation} \label{eq:dim:reset}
\varref{V} > \varref{V}_{\text{peak}}
\quad \text{then} \quad
\begin{cases}
\varref{V} \gets \varref{V}_r, \\
\varref{w} \gets \varref{w} + \varref{b},
\end{cases}
\end{equation}
where
\[
\varref{V}_{\text{peak}} \triangleq \frac{V_{\text{peak}} - V_T}{\Delta_T}, \quad
\varref{V}_r \triangleq \frac{V_r - V_T}{\Delta_T},
\quad \text{and} \quad
\varref{b} \triangleq \frac{b}{\Delta_T g_L}.
\]

\begin{remark}
Equation \eqref{eq:dim:model} and the reset condition \eqref{eq:dim:reset}
will be used from this point onwards. To simplify the notation, we will drop the
overline notation (e.g. $\varref{V}$) in the following with the understanding
that the variables are all non-dimensional.
\end{remark}

\section{Numerical Methods}
\label{sc:methods}

For the discretization of the \FrAdEx{} model \eqref{eq:dim:model}, we use an L1-type
method on a non-uniform grid, to accurately account for the exponential growth
of the solutions and the discontinuous spiking. For a complete description
and analysis of the standard L1 method see classic monographs such
as \cite{li2015numerical}. Following the same ideas, we use a linear approximation
of the function in each interval, i.e. for $s \in [t_n, t_{n + 1}]$ we have that
\[
\vect{y}(s) \approx
    \frac{t_{n + 1} - s}{t_{n + 1} - t_n} \vect{y}_n^+ +
    \frac{s - t_n}{t_{n + 1} - t_n} \vect{y}_{n + 1}^-
\implies
\vect{y}'(s) \approx \frac{\vect{y}^-_{n + 1} - \vect{y}^+_n}{t_{n + 1} - t_n},
\]

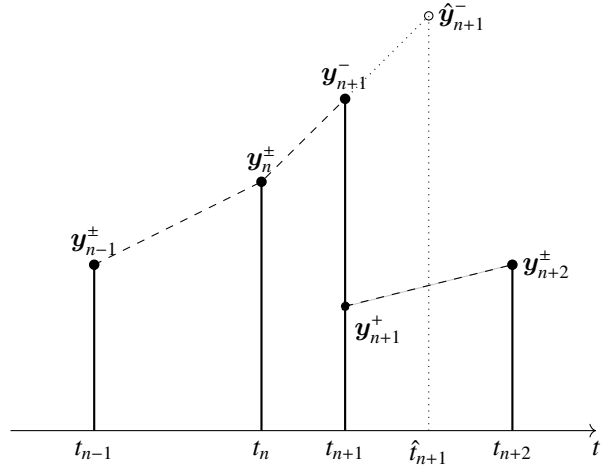
\begin{figure}[ht!]
\centering
\begin{tikzpicture}[scale=1.1]
\draw[->] (0, 0) -- (7, 0) node [below] {$t$};

\draw[thick, fill]
    (1, 0) node [below] {$t_{n - 1}$} --
    (1, 2) circle [radius=0.05] node [above] {$\vect{y}^\pm_{n - 1}$};
\draw[thick, fill]
    (3, 0) node [below] {$t_n$} --
    (3, 3) circle [radius=0.05] node [above] {$\vect{y}^\pm_n$};
\draw[thick, fill]
    (6, 0) node [below] {$t_{n + 2}$} --
    (6, 2) circle [radius=0.05] node [right] {$\vect{y}^\pm_{n + 2}$};
\draw[dotted] (5, 0) node [below] {$\varhat{t}_{n + 1}$} -- (5, 5) -- (4, 4);
\draw (5, 5) circle [radius=0.05] node [right] {$\varhat{\vect{y}}^-_{n + 1}$};
\draw[dashed] (1, 2) -- (3, 3) -- (4, 4);
\draw[dashed, fill] (4, 1.5) node [below right]
    {$\vect{y}^+_{n + 1}$} circle [radius=0.05] -- (6, 2);

\draw[thick, fill]
    (4, 0) node [below] {$t_{n + 1}$} --
    (4, 4) circle [radius=0.05] node [above] {$\vect{y}^-_{n + 1}$};
\end{tikzpicture}
\caption{Reconstruction of the solution at a spike. The initial guess
    $\varhat{\vect{y}}^-_{n + 1}$ is pulled back to the approximated time
    $t_{n + 1}$, where the discontinuous spike occurs.}
\label{fig:methods:spikes}
\end{figure}

\noindent where $\vect{y}_n^\pm$ denote the the right and left-sided limits, respectively,
at $t_n$. The two values are different when a numerical spike has been inserted
at time $t_n$ (see \Cref{fig:methods:spikes}). This results in the following discrete
form of the system \eqref{eq:dim:model}
\begin{equation} \label{eq:methods:l1_naive}
\sum_{k = 0}^n
    \vect{d}_{n + 1, k} \odot \frac{\vect{y}^-_{k + 1} - \vect{y}^+_k}{\Delta t_k} =
    \vect{f}(t_{n + 1}, \vect{y}_{n + 1}^-),
\end{equation}
where $\odot$ denotes the component-wise Hadamard product. The weights are given by
\begin{equation} \label{eq:methods:l1_weights}
(d_{n + 1, k})_i \triangleq
    \frac{(t_{n + 1} - t_k)^{1 - {\alpha}_i} - (t_{n + 1} - t_{k + 1})^{1 - {\alpha_i}}}
        {\Gamma(2 - \alpha_i)}.
\end{equation}

By construction, the method is implicit and the solution $\vect{y}^-_{n + 1}$
must be obtained by finding the root of the nonlinear equation
\[
\vect{y}^-_{n + 1}
- \vect{h}_{n + 1} \odot \vect{f}(t_{n + 1}, \vect{y}_{n + 1})
=
    \vect{y}_n^+
    - \vect{h}_{n + 1} \odot \sum_{k = 0}^{n - 1}
        \vect{d}_{n + 1, k} \odot \frac{\vect{y}^-_{k + 1} - \vect{y}_k^+}{\Delta t_k},
\]
where $(h_{n + 1})_i = \Delta t_n / (d_{n + 1, n})_i =
\Gamma(2 - \alpha_i) \Delta t_n^{\alpha_i}$.

However, as is, the discretization does not yet take into account the reset
condition \eqref{eq:dim:reset}. In practical terms, in \eqref{eq:methods:l1_naive}
we can encounter a reset in the current time interval $[t_n, t_{n + 1}]$. To
maintain the requirement that the numerical spike times occur at interval boundaries, we
recompute $t_{n + 1}$ to ensure that the solution remains piecewise
continuous, as shown in \Cref{fig:methods:spikes}. Therefore, both $y_{n + 1}^-$ and
$t_{n + 1}$ are provisional in this equation and may be modified by the end of
the time step. To emphasise this fact, we rewrite the equation as
\begin{equation} \label{eq:methods:l1}
\varhat{\vect{y}}^-_{n + 1} -
\varhat{\vect{h}}_{n + 1} \odot \vect{f}(\varhat{t}_{n + 1}, \varhat{\vect{y}}^-_{n + 1})
= \vect{y}^+_n - \varhat{\vect{h}}_{n + 1} \odot \sum_{k = 0}^{n - 1}
        \hat{\vect{d}}_{n + 1, k}
        \odot \frac{\vect{y}^-_{k + 1} - \vect{y}_k^+}{\Delta t_k}
= \hat{\vect{r}}_{n + 1},
\end{equation}
where the hat terms denote a dependence on the yet-to-be-determined
$\varhat{t}_{n + 1}$. The complete pseudocode of the
algorithm that evolves the solution to the next time step is given
in \Cref{alg:methods:evolution}.

\begin{algorithm}[H]
\caption{L1 method for \FrAdEx{} integrate-and-fire models.}
\label{alg:methods:evolution}
\begin{algorithmic}
\REQUIRE{Derivative order $\vect{\alpha}$ and time span $[0, T]$.}
\REQUIRE{I.C. $\vect{y}_0$, parameters $(I, E_L, \tau_w, a)$
    from \eqref{eq:dim:model}, and $(V_{\text{peak}}, V_r, b)$
    from \eqref{eq:dim:reset}.}
\REQUIRE{Initial step $\Delta t_0$ and adaptive algorithm parameters
    from \Cref{ssc:methods:adaptive}.}

\WHILE{$t_n \le T$}
    \STATE \begin{enumerate}
        \item[1.] Solve \eqref{eq:methods:l1} using methods
        from \Cref{ssc:methods:lambertw}.

        \item[2.] If $\varhat{V}^-_{n + 1}$ is real, continue to step 4 with
        $t_{n + 1} \equiv \varhat{t}_{n + 1}$ and
        \[
        (\vect{y}^-_{n + 1}, \vect{y}^+_{n + 1}) =
        (\varhat{\vect{y}}^-_{n + 1}, \varhat{\vect{y}}^-_{n + 1}).
        \]
        \item[3.1] If $\varhat{V}^-_{n + 1}$ is complex, compute a step
        according to \Cref{ssc:methods:spikes} and set
        \[
        t_{n + 1} = t_n + \Delta t_{\text{Lambert}}.
        \]
        \item[3.2] If $t_{n + 1} - t_n \le \Delta t_{\text{min}}$, assume a spike occurred
        and set
        \[
        \begin{aligned}
        (V^-_{n + 1}, V^+_{n + 1}) & = (V_{\text{peak}}, V_r), \\
        (w^-_{n + 1}, w^+_{n + 1}) & = (c_0 V_{\text{peak}} + c_1, c_0 V_{\text{peak}} + c_1 + b),
        \end{aligned}
        \]
        where $c_0$ and $c_1$ are defined in \eqref{eq:methods:lambert}. Continue
        to step 5 (step accepted).

        \item[4.] $h \gets$
        \texttt{DetermineTimestep}{$(t_0, \vect{y}^\pm_0), \dots, (t_{n + 1}, \vect{y}^\pm_{n + 1})$}
        using \Cref{ssc:methods:adaptive}.
        \item[5.] If the step is accepted, continue to the the next step
        with $\Delta t_{n + 1} = h$; \\
        otherwise retry the step with $\Delta t_n = h$ \emph{without} updating
        the solution.
    \end{enumerate}
\ENDWHILE
\end{algorithmic}
\end{algorithm}

In \Cref{ssc:methods:lambertw}, we describe a solution method for \eqref{eq:methods:l1}
that does not require an iterative solver, but instead relies on computing the
special Lambert W function. The Lambert W function can also be used to bound the
allowable time step $\Delta t_{n + 1}$, which is described in \Cref{ssc:methods:spikes}.
Finally, an adaptive time-stepping algorithm is described in \Cref{ssc:methods:adaptive}.

\subsection{Lambert W Solution}
\label{ssc:methods:lambertw}

The evolution equation \eqref{eq:methods:l1} is implicit and must be
solved for the value of $\vect{y}^-_{n + 1}$ (and $\vect{y}^+_{n + 1}$ according
to \Cref{alg:methods:evolution}). While this can be done by iterative Newton--Raphson-type
methods, it is not necessarily for this class of adaptive exponential
integrate-and-fire models, as an analytical solution can be obtained. The equation
can be written component-wise, while expanding the right-hand side
from \eqref{eq:dim:model},
\[
\begin{aligned}
\varhat{V}^- & = \varhat{h}_V \left(
    I
    - (\varhat{V}^- - E_L)
    + \exp(\varhat{V}^-)
    - \varhat{w}^-
\right) + \varhat{r}_V, \\
\varhat{w}^- & = \frac{\varhat{h}_w}{\tau_w} \left(
    a (\varhat{V}^- - E_L) - \varhat{w}^-
\right) + \varhat{r}_w,
\end{aligned}
\]
where we have dropped the $n + 1$ index to simplify the notation. We can see
here that the second equation can be solved for $\varhat{w}^-$ for a known
$\varhat{V}^-$. We have that
\[
\varhat{w}^-
= \frac{a \varhat{h}_w}{\varhat{h}_w + \tau_w} \varhat{V}^-
+ \frac{\tau_w \varhat{r}_w - a \varhat{h}_w E_L}{\varhat{h}_w + \tau_w}
\triangleq c_0 \varhat{V}^- + c_1,
\]
which we can then replace into the $\varhat{V}^-$ equation to obtain
\[
\hat{V}^- + c_2 = c_3 \exp (\hat{V}^-),
\]
where
\[
c_2 = \frac{\varhat{h}_V (I + E_L - c_1) + \varhat{r}_V}
            {1 + \varhat{h}_V (1 + c_0)}
\qquad \text{and} \qquad
c_3 = \frac{\varhat{h}_V}{1 + \varhat{h}_V (1 + c_0)}.
\]

This equation has a known analytical solution in terms of the Lambert W function.
Therefore, the final solution is written as
\begin{equation} \label{eq:methods:lambert}
\begin{cases}
\varhat{V}^- = -c_2 - W[-c_3 \exp(-c_2)], \\
\varhat{w}^- = c_0 \varhat{V}^- + c_1.
\end{cases}
\end{equation}

\subsection{Time Step Limit}
\label{ssc:methods:spikes}

The Lambert W-based solution \eqref{eq:methods:lambert} also provides an accurate
method to approximate the maximum allowable time step when approaching
a spike. By definition, we know that the Lambert W function gives a complex-valued
solution to $y = x e^x$. If we require the solution to be real, then
\[
x =
\begin{cases}
W_0(y), & \quad ~ \,0 \le y, \\
W_0(y), W_{-1}(y), & \quad -\frac{1}{e} \le y < 0,
\end{cases}
\]
where $W_k(y)$ denotes the $k$-th branch of the Lambert W function. We can see
that on $-1/e \le y < 0$ the solution is multivalued, with the property that
\[
W_0(y) \ge W_{-1}(y),
\qquad \forall y \in \left(-\frac{1}{e}, 0\right),
\]
so the $k = 0$ branch is always larger on the given interval. When considering
the two branches in \eqref{eq:methods:lambert}, we always choose the $k = 0$ branch,
as it results in a smaller change in $\hat{V}^-$ and is consistent in the
limit of $\Delta t_n \to 0$.

Therefore, ensuring that the argument of the Lambert W function is at most equal
to the lower bound $-1/e$ will result in an equivalent bound on the allowable
time step $\Delta t_n \le \Delta t_{\text{Lambert}}$. Following \eqref{eq:methods:lambert},
we require that
\[
-\frac{1}{e} \le -c_3 \exp\left(-c_2\right) < 0.
\]

\begin{remark}
We have that
\[
c_3 = \frac{\varhat{h}_V}{1 + \varhat{h}_V (1 + c_0)} > 0
\Longleftrightarrow
a \ge -\frac{\varhat{h}_w + \tau_w}{\varhat{h}_w} \frac{1 + \varhat{h}_V}{\varhat{h}_V},
\]
where $\varhat{h}_V, \varhat{h}_w, \tau_w > 0$. In the limit of small $\Delta t_n$,
the inequality is satisfied for all parameters $(a, \tau_w)$ of physical
significance (see \cite[Table 1]{naud2008firing}), so we assume $c_3 > 0$.
\end{remark}

Assuming that $c_3 > 0$, the condition can be simplified to
\begin{equation} \label{eq:methods:spike_time_limits}
0 < c_3 \exp(-c_2 + 1) \le 1,
\end{equation}
where the left inequality is satisfied automatically and the right inequality
can be used to estimate a maximum allowable time step. Then, to obtain this time
step we must find the root of the following nonlinear equation
\[
c_3(\Delta t^*) \exp(-c_2(\Delta t^*) + 1) = 1,
\]
for $\Delta t^* \in [0, \hat{t}_{n + 1} - t_n]$. However, this
can be very inefficient due to the fact that the time step $\Delta t^*$ appears
in the memory terms $\varhat{\vect{r}}_{n + 1}$ \eqref{eq:methods:l1}, which
would need to be recomputed at each iteration of the optimisation problem. In
practice, we consider that the memory terms do not depend explicitly on the time
step. This gives an upper bound on the time step that is not optimal in the
sense of \eqref{eq:methods:spike_time_limits}, but gives sufficiently good results.
Then, the Lambert W-based maximum time step estimate is set to
\[
    \Delta t_{\text{Lambert}} = \Delta t^*.
\]

\begin{remark}
The optimisation problem required to find $\Delta t^*$ can be solved efficiently.
The argument of the Lambert W function is well behaved as a function of
$\Delta t^*$ and we know that the solution is tightly bracketed in
$[0, \hat{t}_{n + 1} - t_n]$.
\end{remark}

\begin{remark}
If the assumption that the memory terms do not depend on the time step is seen
as prohibitive, the algorithm can always fall back to solely making use of the
adaptive algorithm described in \Cref{ssc:methods:adaptive}. This may result in
a poorer approximation of the spike times, but the method itself maintains the
error estimates provided in \Cref{sc:error}.
\end{remark}

\subsection{Adaptive Time Step}
\label{ssc:methods:adaptive}

Due to the exponential growth near the generation of an action potential, as
seen in the results from \cite{naud2008firing}, an accurate solution
to \eqref{eq:dim:model} requires an adaptive time-stepping method. Adaptive
methods are well developed for the integer order case \cite{Hairer2010},
where the truncation error can be used to give accurate approximations of the
required time steps. However, equivalent methods have not yet been developed
for fractional-order evolution equations. A notable exception is the recent work
from \cite{Jannelli2020}, which introduces a simple error indicator with good
properties for fractional equations. Here, we give a short description of the
adaptive algorithm.

The error estimator derived in \cite{Jannelli2020} is given by
\[
\chi_{n + 1} \triangleq
    \Gamma(1 + \alpha) \frac{t_{n + 1} - t_n}{t^\alpha_{n + 1} - t_n^\alpha}
    \frac{\|\vect{y}_{n + 1} - \vect{y}_n\|}{\|\vect{y}_n\|},
\]
for a chosen norm $\|\cdot\|$. The error $\chi_{n + 1}$ is expected to be
maintained between $\chi_{\text{min}} < \chi_{n + 1} < \chi_{\text{max}}$. As
written, this error estimator has a few downsides: it has units of $s^{1 - \alpha}$
and does not apply directly to the case of different orders $\alpha$ for each
component. Therefore, we propose a simple extension in the form of
\begin{equation} \label{eq:methods:estimator}
\hat{\chi}_{n + 1} \triangleq \|\Gamma(1 + \vect{\alpha})\|
    \frac{(t_{n + 1} - t_n)^\alpha}{t_{n + 1}^\alpha - t_n^\alpha}
    \frac{\|\vect{y}_{n + 1} - \vect{y}_n\|}{\|\vect{y}_n\|}
\implies
\chi_{n + 1} \triangleq \frac{\hat{\chi}_{n + 1} - \chi_{\text{min}}}
                             {\chi_{\text{max}} - \chi_{\text{min}}}.
\end{equation}

The new $\chi_{n + 1}$ estimator has normalised values in $[0, 1]$, based on the
reference values $\chi_{\text{min}}$ and $\chi_{\text{max}}$, and is non-dimensional.
Then, the adaptive algorithm follows the steps from \cite{Jannelli2020}:
\begin{enumerate}
    \item Compute the estimator $\chi_{n + 1}$ from \eqref{eq:methods:estimator}.
    \item If $0 < \chi_{n + 1} < 1$, set $\Delta t_{n + 1} = \theta \Delta t_n$,
    for a safety factor $\theta \in (0, 1]$ close to $1$.
    \item If $\chi_{n + 1} < 0$, the step is accepted and can be increased
    by setting $\Delta t_{n + 1} = \rho \Delta t_n$, for $\rho > 1$.
    \item If $\chi_{n + 1} > 1$, the step is rejected and must be decreased
    by setting $\Delta t_{n + 1} = \sigma \Delta t_n$, for $\sigma \in (0, 1]$,
    subject to the requirement in step 5.
    \item If $\Delta t_{n + 1} < \Delta t_{\text{min}}$, the step is accepted
    and we set $\Delta t_{n + 1} = \Delta t_{\text{min}}$.
\end{enumerate}

As such, the algorithm has a total of $6$ parameters that must be chosen based
on the problem at hand. The $\chi_{\text{min}}$ and $\chi_{\text{max}}$ bounds
are the most important, as they determine the thresholds where the time step can
be increased or decreased. The analysis from \cite{Jannelli2020} does not provide
a robust way to choose these parameters and we do not explore this improvement.

\section{Stability and Convergence}
\label{sc:error}

In this section, we present an error analysis of the piecewise L1 method
described in \Cref{sc:methods}, which can be applied to the \FrAdEx{} model and other
similar impulsive models. For simplicity, we only consider the scalar case,
as the extension to the vector case is straightforward. The scalar equation is
\begin{equation} \label{eq:error:model}
\od{^\alpha y}{t^\alpha} = f(t, y)
\quad \text{if} \quad y > y_{\text{peak}} \quad \text{then} \quad y \gets y_r,
\end{equation}
where the right-hand side is assumed to be nonlinear and $0 < \alpha < 1$. The
discretization from \Cref{sc:methods} applies directly to this equation and can
be written as
\begin{equation} \label{eq:error:discrete}
y^-_{n + 1}
- h_{n + 1} f(t_{n + 1}, y_{n + 1})
=
y_n^+
- h_{n + 1} \sum_{k = 0}^{n - 1}
    d_{n + 1, k} \frac{y^-_{k + 1} - y_k^+}{\Delta t_k}.
\end{equation}

We recall that the standard L1 method has an order of $\mathcal{O}(h^{2 - \alpha})$
on a uniform mesh of spacing $h$ for sufficiently smooth functions. On a
non-uniform mesh, the analysis is provided by \cite[Theorem 4.2]{li2019theory}
with $y \in C^2([0, 2])$. An extension of the results from \cite{li2019theory}
to the case of impulsive systems is trivial when the spike times $\{\tau_m\}$ are
known a priori and the grid $\{t_n\}$ incorporates the discontinuities.
We first define
\[
\Delta t_{\text{max}} \triangleq \max_{0 \le n < N} \Delta t_n
\quad \text{and} \quad
\Delta t_{\text{min}} \triangleq \min_{0 \le n < N} \Delta t_n,
\]
and assume that $C_{\Delta t} = \Delta t_{\text{max}}/{\Delta t_{\text{min}}}$ is
finite, but can be quite large in practice (e.g. due to the exponential
adaptation of \eqref{eq:model:fde}). We also denote by $y(t_n^\pm)$ and $y_n^\pm$
the known continuous solution and the numerical solution, respectively.

\begin{lemma}
Let $0 < \alpha < 1$ and $y \in \mathrm{PAC}^2([0, T])$, with a countable set
of discontinuities $\tau_1 < \cdots < \tau_m < \cdots < \tau_M$. Let the numerical
grid $0 = t_0 < \cdots < t_n < \cdots t_{N} = T$ be such that for every $m$ there
exists an $n$ such that $\tau_m = t_n$. Then, for every $n$, it holds that
\begin{equation}
\od{^\alpha y}{t^\alpha}(t_{n + 1}) =
\sum_{k = 0}^n d_{n + 1, k}
    \frac{y(t_{k + 1}^-) - y(t_k^+)}{\Delta t_k} + R_{n + 1},
\end{equation}
and
\begin{equation}
|R_{n + 1}|\leq \frac{1}{\Gamma(1 - \alpha)} \left[
\frac{\Delta t_n^{2 - \alpha}}{2 (1 - \alpha)}
+ \frac{\Delta t_{\text{max}}^2 \Delta t_n^{-\alpha}}{8}
\right] \esssup_{0 \leq s \leq t_{n + 1}}|y''(s)|.
\end{equation}
\end{lemma}

\begin{proof}
See proof of \cite[Theorem 4.2]{li2019theory}.

\qed
\end{proof}

We now consider the case of interest, where the spike times $\{\tau_m\}$ are
solution-dependent. Then, a grid $\{t_k\}$ that conforms to the piecewise
nature of the solution cannot be constructed a priori. However, a similar
error estimate can be constructed for this class of impulsive systems. We will
assume that the numerical spike times are a first-order approximation of the
real spike times $\{\tau_m\}$, to match the construction from \Cref{sc:methods}. The
construction of higher-order methods for models with exponential growth is left
for future study.

We start by giving an estimate for truncation error below.

\begin{theorem}[Piecewise L1 Truncation Error] \label{thm:error:truncation_error}
Let $0 < \alpha < 1$ and $y \in \mathrm{PAC}^2([0, T])$, with a countable set
of discontinuities $\tau_1 < \cdots < \tau_m < \cdots < \tau_M$. Let the numerical
grid $0 = t_0 < \cdots < t_n < \cdots < t_{N} = T$ be such that for every $m$
there exists an $n$ such that $\tau_m \in [t_{n - 1}, t_{n + 1}]$. Then, for
$t_{n + 1} \in [\tau_m, \tau_{m + 1}]$, it holds that
\begin{equation}
\begin{aligned}
\od{^\alpha y}{t^\alpha}(t_{n + 1}) & =
\sum_{k = 0}^n d_{n + 1, k} \frac{y(t_{k + 1}^-) - y(t_k^+)}{t_{k + 1} - t_k} \\
& - \frac{1}{\Gamma(1 - \alpha)}
  \sum_{j = 1}^{m - 1} \frac{y(\tau_j^+) - y(\tau_j^-)}{(t_{n + 1} - \tau_j)^\alpha}
+ R_{n + 1},
\end{aligned}
\end{equation}
where
\begin{equation}
|R_{n + 1}|\leq \frac{1}{\Gamma(1 - \alpha)} \left(
c_1 \Delta t_n^{1 - \alpha} + c_2 \Delta t_{\text{max}} \Delta t_n^{-\alpha}
\right).
\end{equation}
\end{theorem}

\begin{proof}
The derivation of the truncation error for the case where the spike times are not
known follows ideas similar to the proof of \cite[Theorem 4.2]{li2019theory}. We
show here the proof and highlight the differences. From \Cref{sc:fractional}, we
have that
\[
\od{^\alpha y}{t^\alpha}(t_{n + 1}) =
\frac{1}{\Gamma(1 - \alpha)} \left[
\sum_{j = 0}^{m - 1}
    \int_{\tau_j}^{\tau_{j + 1}} \frac{y'(s)}{(t_{n + 1} - s)^{\alpha}} \dx[s]
    + \int_{\tau_m}^{t_{n + 1}} \frac{y'(s)}{(t_{n + 1} - s)^{\alpha}} \dx[s]
\right],
\]
for $t_{n + 1} \in (\tau_m, \tau_{m + 1}]$. For the current time step,
we must separately consider the case where $\tau_m \in [t_n, t_{n + 1}]$ and
$\tau_m \notin [t_n, t_{n + 1}]$. For simplicity, we only expand on the first
case, where there is a discontinuity in the current interval. Then, we have that
\[
\begin{aligned}
& \sum_{j = 0}^{m - 1}
    \int_{\tau_j}^{\tau_{j + 1}} \frac{y'(s)}{(t_{n + 1} - s)^{\alpha}} \dx[s]
    + \int_{\tau_m}^{t_{n + 1}} \frac{y'(s)}{(t_{n + 1} - s)^{\alpha}} \dx[s] \\
=\,\, & \sum_{j = 0}^{m - 2}
    \int_{\tau_j}^{\tau_{j + 1}} \frac{y'(s)}{(t_{n + 1} - s)^{\alpha}} \dx[s]
    + \int_{\tau_{m - 1}}^{t_n} \frac{y'(s)}{(t_{n + 1} - s)^{\alpha}} \dx[s] \\
+\,\, & \int_{t_n}^{\tau_m} \frac{y'(s)}{(t_{n + 1} - s)^{\alpha}} \dx[s]
    + \int_{\tau_m}^{t_{n + 1}} \frac{y'(s)}{(t_{n + 1} - s)^{\alpha}} \dx[s].
\end{aligned}
\]

On $[t_n, t_{n + 1}]$ we can use a simple Taylor expansion to obtain a bound for
the derivative. Due to the discontinuity, the expansion can only be $\mathcal{O}(1)$,
which gives
\[
\begin{aligned}
& \int_{t_n}^{\tau_m} \frac{y'(s)}{(t_{n + 1} - s)^{\alpha}} \dx[s]
+ \int_{\tau_m}^{t_{n + 1}} \frac{y'(s)}{(t_{n + 1} - s)^{\alpha}} \dx[s] \\
=\,\, & \tilde{d}_{n + 1, n} \frac{y(t_{n + 1}^-) - y(t_n^+)}{t_{n + 1} - t_n}
+ \int_{t_n}^{t_{n + 1}} \frac{R_{n, n + 1}(s)}{(t_{n + 1} - s)^\alpha} \dx[s],
\end{aligned}
\]
where $\tilde{d}_{n + 1, n} = \Gamma(1 - \alpha) d_{n + 1, n}$ and $|R_{n, n + 1}'|
< c$. If $\tau_m < t_n$, then the estimate $|R_{n, n + 1}| < c \Delta t_n$
from \cite{li2019theory} should be used. For the remaining intervals, we use the
standard linear Lagrange interpolator of the L1 method. If the interval
$[t_k, t_{k + 1}]$ contains a discontinuity, the interpolation error is given
by
\[
y(s) =
    \frac{t_{k + 1} - s}{t_{k + 1} - t_k} y(t_k)
    + \frac{s - t_k}{t_{k + 1} - t_k} y(t_{k + 1})
    + R_{k, k + 1}(\xi),
\]
for $\xi \in [t_k, t_{k + 1}]$, where $|R_{k, k + 1}| < c \Delta t_k$. A proof
of this bound is found, for example, in \cite[Lemma 5]{Nguyen2009}. If the
solution is continuous in the interval, then the bound from \cite{li2019theory}
applies again. Introducing this estimate into the integral, we have that
\[
\begin{aligned}
\alpha \int_{t_k}^{t_{k + 1}}
    \frac{y(s)}{(t_{n + 1} - s)^{1 + \alpha}} \dx[s] & =
\left.\frac{y(s)}{(t_{n + 1} - s)^\alpha}\right|_{t_k}^{t_{k + 1}}
- \tilde{d}_{n + 1, k} \frac{y(t_{k + 1}^-) - y(t_k^+)}{t_{k + 1} - t_k} \\
&
+ \int_{t_k}^{t_{k + 1}} \!\!\!\!
    \frac{R_{k, k + 1}}{(t_{n + 1} - s)^{1 + \alpha}} \dx[s].
\end{aligned}
\]

This expression can be used to give an error estimate for the remaining terms.
We have, by integration by parts, that
\[
\begin{aligned}
& \sum_{j = 0}^{m - 2}
    \int_{\tau_j}^{\tau_{j + 1}} \frac{y'(s)}{(t_{n + 1} - s)^{\alpha}} \dx[s]
    + \int_{\tau_{m - 1}}^{t_n} \frac{y'(s)}{(t_{n + 1} - s)^{\alpha}} \dx[s] \\
=\,\, &
\sum_{j = 0}^{m - 2} \left[
\left. \frac{y(s)}{(t_{n + 1} - s)^\alpha}\right|_{\tau_j}^{\tau_{j + 1}}
- \alpha \int_{\tau_j}^{\tau_{j + 1}} \frac{y(s)}{(t_{n + 1} - s)^{1 + \alpha}} \dx[s]
\right] \\
+\,\, &
\left. \frac{y(s)}{(t_{n + 1} - s)^\alpha}\right|_{\tau_{m - 1}}^{t_n}
- \alpha \int_{\tau_{m - 1}}^{t_n} \frac{y(s)}{(t_{n + 1} - s)^{1 + \alpha}} \dx[s] \\
=\,\, &
\left. \frac{y(s)}{(t_{n + 1} - s)^\alpha}\right|_{\tau_0}^{t_n}
- \sum_{j = 1}^{m - 1} \frac{y(\tau_j^+) - y(\tau_j^-)}{(t_{n + 1} - \tau_j)^\alpha}
- \alpha \int_{\tau_0}^{t_n} \frac{y(s)}{(t_{n + 1} - s)^{1 + \alpha}} \dx[s] \\
=\,\, &
\sum_{k = 0}^{n - 1} \tilde{d}_{n + 1, k} \frac{y(t_{k + 1}^-) - y(t_k^+)}{t_{k + 1} - t_k}
- \sum_{j = 1}^{m - 1} \frac{y(\tau_j^+) - y(\tau_j^-)}{(t_{n + 1} - \tau_j)^\alpha}
- \sum_{k = 0}^{n - 1}
    \int_{t_k}^{t_{k + 1}} \frac{R_{k, k + 1}(s)}{(t_{n + 1} - s)^{1 + \alpha}} \dx[s].
\end{aligned}
\]

The remainder can be determined by putting the results on $[0, t_n]$ and
$[t_n, t_{n + 1}]$ together to give
\[
\begin{aligned}
|R_{n + 1}| & =
\frac{1}{\Gamma(1 - \alpha)} \left|
\int_{t_n}^{t_{n + 1}} \frac{R_{n, n + 1}(s)}{(t_{n + 1} - s)^\alpha} \dx[s]
- \sum_{k = 0}^{n - 1} \int_{t_k}^{t_{k + 1}}
    \frac{R_{k, k + 1}(s)}{(t_{n + 1} - s)^{1 + \alpha}} \dx[s]
\right| \\
& \le \frac{1}{\Gamma(1 - \alpha)}
    (c_1 \Delta t_n^{1 - \alpha} + c_2 \Delta t_{\text{max}} \Delta t_n^{-\alpha}).
\end{aligned}
\]
\end{proof}

As we have seen in \Cref{thm:error:truncation_error}, the truncation error
for the method is $\mathcal{O}(1)$ and, thus, non-convergent. This is
not surprising, due to the fact that the grid $\{t_n\}$ does not match the
spike times $\{\tau_m\}$. However, the numerical convergence of the method is
expected to be first order. We define the error as $e_n^\pm = y(t_n^\pm) - y_n^\pm$
and provide an estimate below.

\begin{theorem}[Piecewise L1 Global Error] \label{thm:error:estimate}
Let the right-hand side $f(t, y)$ from \eqref{eq:error:model} be Lipschitz
continuous in its second argument, i.e.
\[
|f(t, x) - f(t, y)| < L |x - y|,
\]
where $L > 0$ is the Lipschitz constant. Let $\{y_n^\pm\}$ be the numerical
solution of \eqref{eq:error:discrete} on a grid $0 = t_0 < \cdots < t_n < \cdots
< t_N = T$ such that $\Gamma(2 - \alpha) L \Delta t_n^\alpha < 1$ for every $n$.
Then, the error satisfies
\[
|e^-_{n + 1}| < c \Delta t_{\text{max}}.
\]
\end{theorem}

\begin{proof}
To obtain an expression for the error, we subtract the residual equation
from \Cref{thm:error:truncation_error} from the discretization \eqref{eq:methods:l1_naive}.
This gives
\[
\begin{aligned}
& e^-_{n + 1} - h_{n + 1} [f(t_{n + 1}, y(t_{n + 1}^-)) - f(t_{n + 1}, y_{n + 1}^-)] \\
=\,\, & e^+_n
- h_{n + 1} \sum_{k = 0}^{n - 1} d_{n + 1, k} \frac{e^-_{k + 1} - e^+_k}{\Delta t_k}
+ \frac{h_{n + 1}}{\Gamma(1 - \alpha)}
    \sum_{j = 1}^{m - 1} \frac{y(\tau_j^+) - y(\tau_j^-)}{(t_{n + 1} - \tau_j)^\alpha}
- h_{n + 1} R_{n + 1}, \\
=\,\, &
\left[
1 - \frac{h_{n + 1} d_{n + 1, n - 1}}{\Delta t_{n - 1}}
\right] e_n^-
- h_{n + 1} \sum_{k = 1}^{n - 1} \left[
\frac{d_{n + 1, k - 1}}{\Delta t_{k - 1}} e_k^-
- \frac{d_{n + 1, k}}{\Delta t_k} e_k^+
\right] \\
+\,\, &
\frac{h_{n + 1}}{\Gamma(1 - \alpha)}
\sum_{j = 1}^{m - 1} \frac{y(\tau_j^+) - y(\tau_j^-)}{(t_{n + 1} - \tau_j)^\alpha}
- h_{n + 1} R_{n + 1}
\end{aligned}
\]

\begin{figure}[ht!]
\centering
\begin{tikzpicture}[scale=0.9]
\draw[->] (0, -1) -- (9, -1);
\draw[fill] (1, -1) circle [radius=0.05] node [below] {$t_{k - 1}$};
\draw[fill] (4, -1) circle [radius=0.05] node [below] {$t_k$};
\draw[fill] (7, -1) circle [radius=0.05] node [below] {$t_{k + 1}$};
\draw[fill] (6, -1) circle [radius=0.05] node [below] {$\tau_{j_k}$};

\draw[thick] (1, 0) -- (4, 1.25) -- (7, 4);
\draw[thick] (7, 1) -- (8, 2);
\draw[dashed,fill]
    (1, -1)
    -- (1, 0) circle [radius=0.05] node [above] {$y_{k - 1}^\pm$};
\draw[dashed,fill]
    (4, -1)
    -- (4, 1.25) circle [radius=0.05] node [above] {$y_k^\pm$};
\draw[dashed,fill]
    (7, -1)
    -- (7, 1) circle [radius=0.05] node [left] {$y_{k + 1}^+$}
    -- (7, 4) circle [radius=0.05] node [right] {$y_{k + 1}^-$};

\draw[scale=0.5, domain=2:12, smooth, variable=\t, gray] plot ({\t}, {(\t*\t-24)/24});
\draw[scale=0.5, domain=12:16, smooth, variable=\t, gray] plot ({\t}, {(\t*\t-150)/34});
\draw[dashed,fill,gray]
    (6, -1)
    -- (6, -0.09) circle [radius=0.05] node [left] {$y(\tau_j^+)$}
    -- (6, 2.5) circle [radius=0.05] node [right] {$y(\tau_j^-)$};
\draw[fill, gray] (7, 0.67) circle [radius=0.05] node [right] {$y(t_{k + 1}^\pm)$};
\end{tikzpicture}
\caption{Approximate solution based on a linear approximation (black) and the
exact solution (gray). The approximate solution has a detected jump at $t_{k + 1}$
and the exact solution has the jump at $\tau_j \in [t_{k - 1}, t_{k + 1}]$.}
\label{fig:error:discontinuities}
\end{figure}
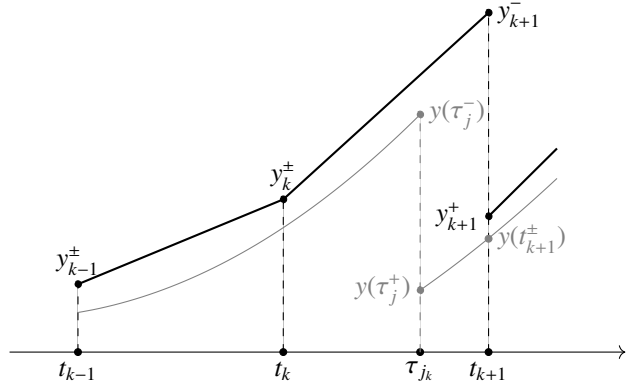

We must now manipulate the errors such that the jump terms cancel out to
first order. In general, given a numerical spike time at $t_k$, we assume that the
real spike time $\tau_{j_k} \in [t_{k - 1}, t_{k + 1}]$, as shown
in \Cref{fig:error:discontinuities}. This can always be achieved if the step
size is sufficiently small. Then, for each numerical spike time $t_k$, we can write
\[
e_k^+ = (y(t_k^+) - y_k^+) = e_k^- - (y(\tau_{j_k}^+) - y(\tau_{j_k}^-)),
\]
and introduce this expression into the error estimate
\[
\begin{aligned}
& h_{n + 1} \sum_{k = 1}^{n - 1} \left[
\frac{d_{n + 1, k - 1}}{\Delta t_{k - 1}} e_k^-
- \frac{d_{n + 1, k}}{\Delta t_k} e_k^+
\right]
- \frac{h_{n + 1}}{\Gamma(1 - \alpha)}
\sum_{j = 1}^{m - 1} \frac{y(\tau_j^+) - y(\tau_j^-)}{(t_{n + 1} - \tau_j)^\alpha} \\
=\,\, &
h_{n + 1} \sum_{k = 1}^{n - 1} \left[
\frac{d_{n + 1, k - 1}}{\Delta t_{k - 1}}
- \frac{d_{n + 1, k}}{\Delta t_k}
\right] e_k^- \\
+\,\, &
h_{n + 1}
\sum_{j = 1}^{m - 1} \frac{d_{n + 1, k_j}}{\Delta t_{k_j}} (y(\tau_j^+) - y(\tau_j^-))
-
\frac{h_{n + 1}}{\Gamma(1 - \alpha)}
\sum_{j = 1}^{m - 1} \frac{y(\tau_j^+) - y(\tau_j^-)}{(t_{n + 1} - \tau_j)^\alpha},
\end{aligned}
\]
where $k_j$ denotes the numerical spike time $t_{k_j}$ corresponding to $\tau_j$,
as denoted above. Now, we consider the weights in the second sum, as
defined by \eqref{eq:methods:l1_weights}, and perform the expansions
\[
\begin{aligned}
(t_{n + 1} - t_k)^{1 - \alpha} & =
(t_{n + 1} - \tau_j)^{1 - \alpha} +
(1 - \alpha) (\tau_j - t_k) (t_{n + 1} - \tau_j)^{-\alpha}
+ \mathcal{O}(\Delta t_k^2), \\
(t_{n + 1} - t_{k + 1})^{1 - \alpha} & =
(t_{n + 1} - \tau_j)^{1 - \alpha}
- (1 - \alpha) (t_{k + 1} - \tau_j) (t_{n + 1} - \tau_j)^{-\alpha}
+ \mathcal{O}(\Delta t_k^2),
\end{aligned}
\]
which give
\[
\frac{d_{n + 1, k}}{\Delta t_k} =
    \frac{1}{\Gamma(1 - \alpha)} \frac{1}{(t_{n + 1} - \tau_j)^\alpha}
    + \mathcal{O}(\Delta t_k).
\]

We can see that the first term will exactly cancel out the original jump terms.
Finally, using \Cref{thm:error:truncation_error}, we have that
\[
\left|\frac{h_{n + 1}}{\Gamma(1 - \alpha)} R_{n + 1}\right| \le
    c_1 \Delta t_n + c_2 \Delta t_{\text{max}}.
\]

Therefore, we can write the error as
\[
\begin{aligned}
& e^-_{n + 1} - h_{n + 1} [f(t_{n + 1}, y(t_{n + 1}^-)) - f(t_{n + 1}, y_{n + 1}^-)] \\
=\,\, &
\left[
1 - \frac{h_{n + 1} d_{n + 1, n - 1}}{\Delta t_{n - 1}}
\right] e_n^-
- h_{n + 1} \sum_{k = 1}^{n - 1} \left[
\frac{d_{n + 1, k - 1}}{\Delta t_{k - 1}}
- \frac{d_{n + 1, k}}{\Delta t_k}
\right] e_k^- \\
-\,\, &
c_0 \sum_{j = 0}^{m - 1} \Delta t_{k_j} (y(\tau_j^+) - y(\tau_j)^-)
+ c_1 \Delta t_n
+ c_2 \Delta t_{\text{max}} \\
=\,\, &
\left[
1 - \frac{h_{n + 1} d_{n + 1, n - 1}}{\Delta t_{n - 1}}
\right] e_n^-
- h_{n + 1} \sum_{k = 1}^{n - 1} \left[
\frac{d_{n + 1, k - 1}}{\Delta t_{k - 1}}
- \frac{d_{n + 1, k}}{\Delta t_k}
\right] e_k^- + c \Delta t_{\text{max}},
\end{aligned}
\]
for some positive constants $c_0, c_1, c_2, c \in \mathbb{R}$. We note that
the last estimate can be invalid if the number of spike times becomes very large,
i.e. if $m \max (y(\tau_j^+) - y(\tau_j^-)) = \mathcal{O}(\Delta t_{\text{max}}^{-1})$.
This can be the case if the solution is evolved for a very large time horizon.

We will now apply a variant of the discrete Grönwall inequality
from \cite[Lemma 3.3.1]{li2015numerical}. For this we first require some
bounds on the coefficients of $e_n^-$. We have that
\[
\begin{aligned}
\left|
    1 - \frac{h_{n + 1} d_{n + 1, n - 1}}{\Delta t_{n - 1}}\right|
& \le 1 + \frac{\Delta t_n^\alpha}{\Delta t_{n - 1}^\alpha}
\le 1 + C_{\Delta t}^\alpha, \\
h_{n + 1} \left|
    \frac{d_{n + 1, k - 1}}{\Delta t_{k - 1}}
    - \frac{d_{n + 1, k}}{\Delta t_k}
\right| & \le \left|
\frac{\Delta t_n^\alpha}{\Delta t_{k - 1}^\alpha} -
\frac{\Delta t_n^\alpha}{\Delta t_k^\alpha}
\right| \le 2 C_{\Delta t}^\alpha.
\end{aligned}
\]

Then, taking the absolute value on both sides and using the above bounds, we
obtain
\[
(1 - L h_{n + 1}) |e_{n + 1}^-| \le
c \Delta t_{\text{max}} + \beta \sum_{k = 0}^n |e_k^-|,
\]
where $\beta \triangleq \max(1 + C_{\Delta t}^\alpha, 2 C_{\Delta t}^\alpha)$.
A direct application of the Grönwall inequality gives the desired result when
$1 - L h_{n + 1} > 0$, as initially assumed.
\end{proof}

\begin{remark}
As shown in the proof of \Cref{thm:error:estimate}, we make very few assumptions
on the form of the reset condition and the method of determining the numerical
spike times $\{t_k\}$. Therefore, the error estimate applies not only to the
\FrAdEx{} discretization from \Cref{sc:methods}, but also to a larger class of
problems. For example, in the next section, we apply the same methods to the
simpler perfect integrate-and-fire and the leaky integrate-and-fire
models, where the spike times are approximated by a direct linear interpolation
between $(t_n, V_n^+)$ and $(\hat{t}_{n + 1}, \hat{V}_{n + 1}^-)$.
\end{remark}

\begin{remark}
It is well-known that solutions to fractional differential equations (usually) have
a singularity at $t = 0$ and we cannot generally assume that
$y(t) \in \mathrm{PAC}^2([0, T])$. More realistic spaces have been analysed
for the L1 method~\cite{Jin2015} and for the L2 method~\cite{Quan2023}. In the
case of the L1 method, the global error degrades to $\mathcal{O}(\Delta t)$ due
to the singularity, as in~\Cref{thm:error:estimate}.
\end{remark}

\section{Numerical Simulations}
\label{sc:results}

We present a series of numerical experiments intended to verify and validate the
numerical method presented in \Cref{sc:methods}. The implementation is carried
out in Python with the aid of the \texttt{numpy} and \texttt{scipy}
version libraries. The numerical method itself is implemented on top of
\texttt{pycaputo}~\cite{pycaputo_0_5_0}, an open source library for fractional
calculus developed by the authors.

In the following, we will make use of the full machinery required to solve the
\FrAdEx{} model using adaptive time stepping and the implicit L1 method. All reported
errors will be relative in the $\ell^2$ norm, i.e.
\[
E(\vect{x}, \vect{x}_{\text{ref}}) \triangleq
 \frac{\|\vect{x} - \vect{x}_{\text{ref}}\|_2}{\|\vect{x}_{\text{ref}}\|_2},
\]
where $\vect{x}$ is the approximate solution and $\vect{x}_{\text{ref}}$ is a known
reference value evaluated on the $\{t_n\}$ grid. The remaining parameters are
problem-dependent and will be stated for each test case. In the text, they are
stated in their dimensional form for comparison to existing literature \cite{naud2008firing},
but the simulations are performed on the non-dimensional system from \Cref{sc:dim}.

\subsection{Convergence on a simple PIF model}
\label{ssc:results:exp1}

We start by applying the methods from \Cref{sc:methods} to a simple perfect
integrate-and-fire (PIF) model and use the results to verify the error estimates
from \Cref{sc:error}. The fractional extension of the PIF model is given by
\[
C \od{^\alpha V}{t^\alpha} = I
\quad \text{if} \quad V < V_{\text{peak}}
\quad \text{else} \quad V \gets V_r,
\]
where $C = \qty{100}{\pico\farad \milli\second^{\alpha - 1}}$ is the fractional
capacitance, $I = \qty{160}{\pico\ampere}$ is a constant current,
$V_{\text{peak}} = \qty{0}{\milli\volt}$ is the membrane threshold potential,
and $V_r = \qty{-48}{\milli\volt}$ is the reset potential. The system is
non-dimensionalized using $I_{\text{ref}} = \qty{20}{\pico\ampere}$ and
$V_{\text{ref}} = \qty{1}{\milli\volt}$. The initial condition is set to a
random value in $[V_r, V_{\text{peak}}]$, fixed for each different order $\alpha$
tested below. This problem has a known solution, which is given by
\[
V(t) = V_0 + m (V_r - V_{\text{peak}}) + \frac{t^{\alpha}}{\Gamma(1 + \alpha)} I,
\]
for $t \in (\tau_m, \tau_{m + 1}]$. The corresponding spike times also have an
explicit expression
\[
\tau_{m + 1} = \sqrt[\alpha]{\Gamma(1 + \alpha)
    \frac{V_{\text{peak}} - V_0 - m (V_r - V_{\text{peak}})}{I}}.
\]

\begin{figure}[ht!]
\begin{subfigure}{0.45\linewidth}
\centering
\includegraphics[width=0.9\textwidth]{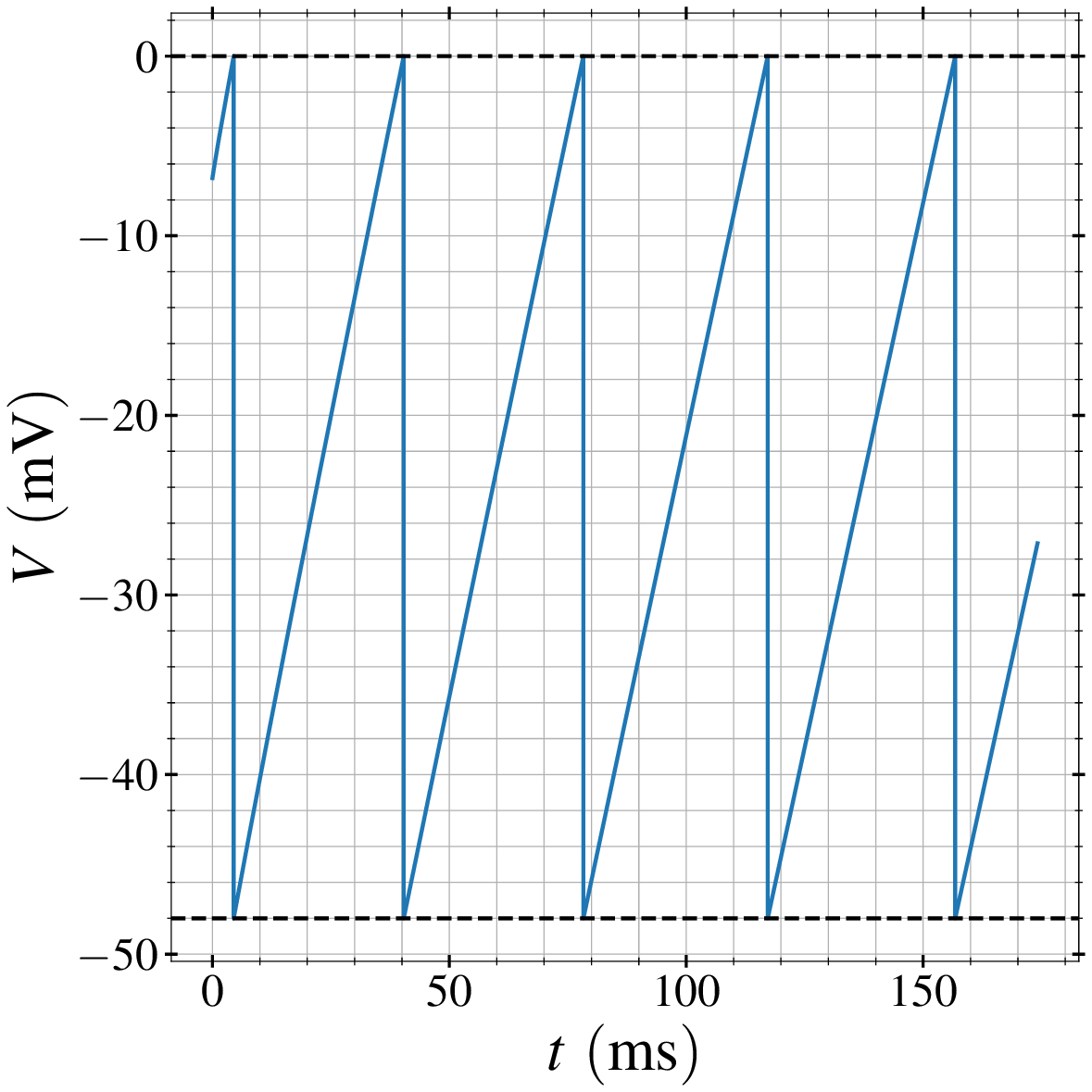}

\caption{Membrane potential at $\alpha = 0.95$.}
\end{subfigure}
\begin{subfigure}{0.45\linewidth}
\centering
\includegraphics[width=0.9\textwidth]{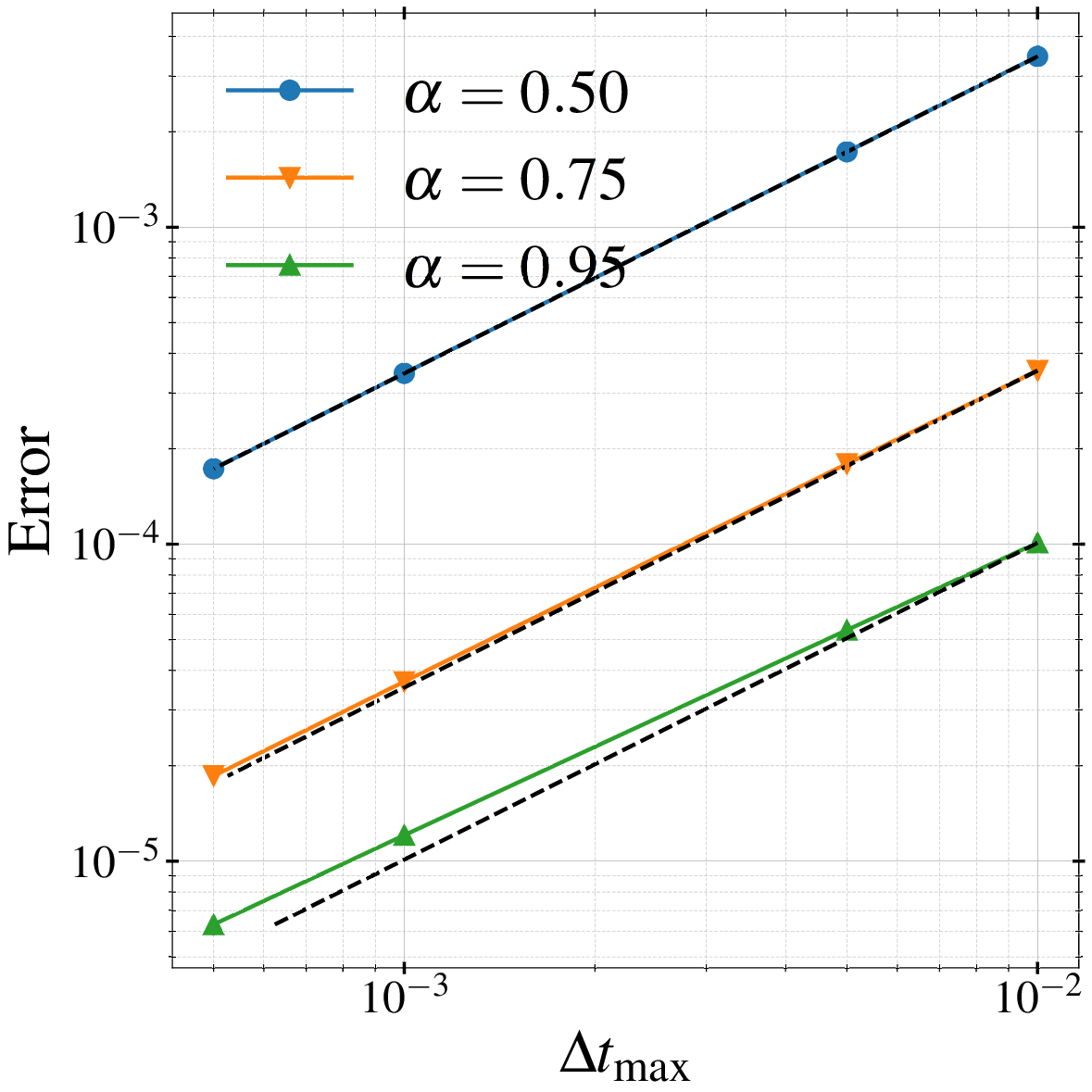}
\caption{Convergence.}
\end{subfigure}

\caption{(a) Membrane potential at $\alpha = 0.95$ and $\Delta t = 5 \times 10^{-4}$.
The peak potential $V_{\text{peak}}$ and the reset potential $V_r$
are denoted with dashed lines. (b) Convergence of the piecewise L1 method on the
PIF model. The error is computed based on the location of the exact and the
approximate spike times $\{\tau_m\}$. The dashed lines denote the expected order
$\mathcal{O}(\Delta t_{\text{max}})$.}
\label{fig:results:exp1}
\end{figure}

This system also falls into the category described in \Cref{sc:error} and the
same error estimates apply. Therefore, we can test the accuracy of the given
estimates for different orders $\alpha$ and different grid sizes. In the case
of the PIF model, we can solve the implicit equation of the L1 method exactly
and do not require the methods from \Cref{ssc:methods:lambertw}. We also replace
the adaptive spike time approximation with a simple linear interpolator: when
$\hat{V}_{n + 1}^- > V_{\text{peak}}$, we set
\[
t_{n + 1} = \frac{V_{\text{peak}} - V_n^+}{\hat{V}_{n + 1}^- - V_n^+} \hat{t}_{n + 1}
    + \frac{\hat{V}_{n + 1}^- - V_{\text{peak}}}{\hat{V}_{n + 1}^- - V_n^+} t_n
\]
and proceed with the reset as described in \Cref{alg:methods:evolution}. To verify the
method, we evolve the equation to a  (non-dimensional) final time $T = 32$
using different $\alpha \in \{0.5, 0.75, 0.95\}$ and fixed time steps
$\Delta t \in \{10^{-2}, 5 \times 10^{-3}, 10^{-3}, 5 \times 10^{-4}\}$. Note
that the PIF model with a constant current $I$ does not require adaptive time
stepping, due to the function being very well-behaved away from $t = 0$. With
these parameters, we have a total of $6$ spike times. We can see
in \Cref{fig:results:exp1} that we obtain the expected $\mathcal{O}(\Delta t)$
order of convergence, due to the first-order estimate of the spike time locations.

\subsection{Adaptivity of a simple LIF model}
\label{ssc:results:exp2}

To briefly validate the step size adaptation presented in \Cref{ssc:methods:adaptive},
we turn to the more complex leaky integrate-and-fire model (LIF). The fractional
extension of the model is given by
\[
C \od{^\alpha V}{t^\alpha} = I - g_L (V - E_L)
\quad \text{if} \quad V < V_{\text{peak}}
\quad \text{else} \quad V \gets V_r,
\]
where $C = \qty{100}{\pico\farad~ \milli\second^{\alpha - 1}}$ is the fractional
capacitance, $I = \qty{160}{\pico\ampere}$ is the constant current,
$g_L = \qty{3}{\nano\siemens}$ is the total leak conductance, $E_L = \qty{-50}{\milli\volt}$
is the effective resting potential, $V_{\text{peak}} = \qty{0}{\milli\volt}$ is
the peak membrane potential, and $V_r = \qty{-48}{\milli\volt}$ is the
reset potential. The system is non-dimensionalized using $V_{\text{ref}} =
\qty{1}{\milli\volt}$ and $\alpha = 0.85$. As before, we evolve the equation to
$T = 32$ (non-dimensional) with a linear approximation of the spike time, but using
the adaptive time stepping method from \Cref{ssc:methods:adaptive}.

For the adaptive algorithm, we set $\theta = 1.0, \sigma = 0.5$ and $\rho = 1.5$.
The minimum time step is set to $\Delta t_{\text{min}} = 10^{-5}$ and the initial
time step is taken to be $\Delta t_0 = 10^{-1}$. The remaining parameters
$\chi_{\text{min}}$ and $\chi_{\text{max}}$ must be determined empirically to
result in a stable and efficient adaptive scheme. In this experiment, we have
chosen them in such a way that adaptivity is showcased for the simple LIF model.
They are
\[
\chi_{\text{min}} = \left\{\frac{2}{2^k} \,\middle|\, k = 0, 4, 8\right\}
\quad \text{and} \quad
\chi_{\text{max}} = \left\{\frac{4}{2^k} \,\middle|\, k = 0, 4, 8\right\}.
\]

\begin{figure}[ht!]
\begin{subfigure}{0.45\linewidth}
\centering
\includegraphics[width=0.9\textwidth]{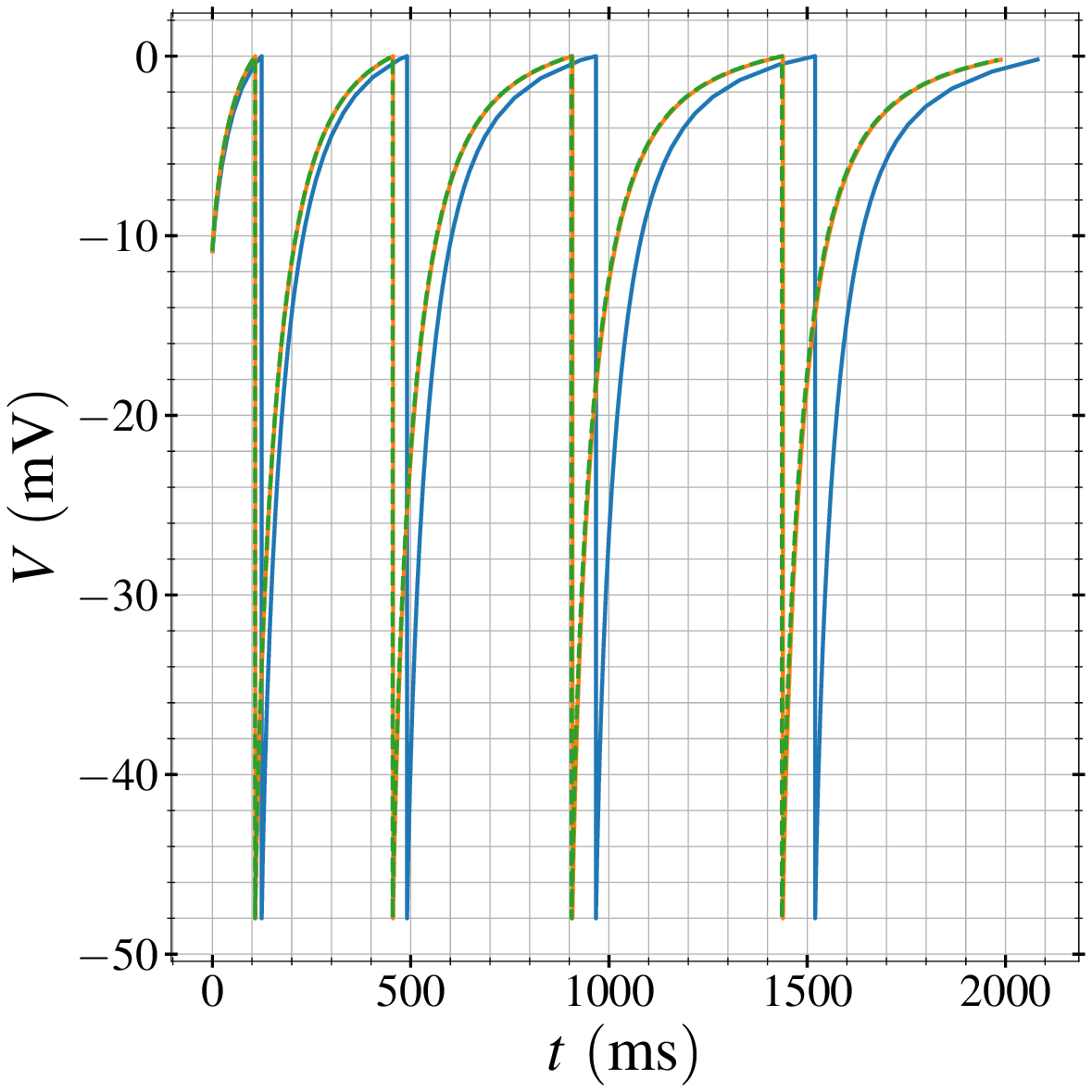}

\caption{Membrane potential $V$.}
\end{subfigure}
\begin{subfigure}{0.45\linewidth}
\centering
\includegraphics[width=0.9\textwidth]{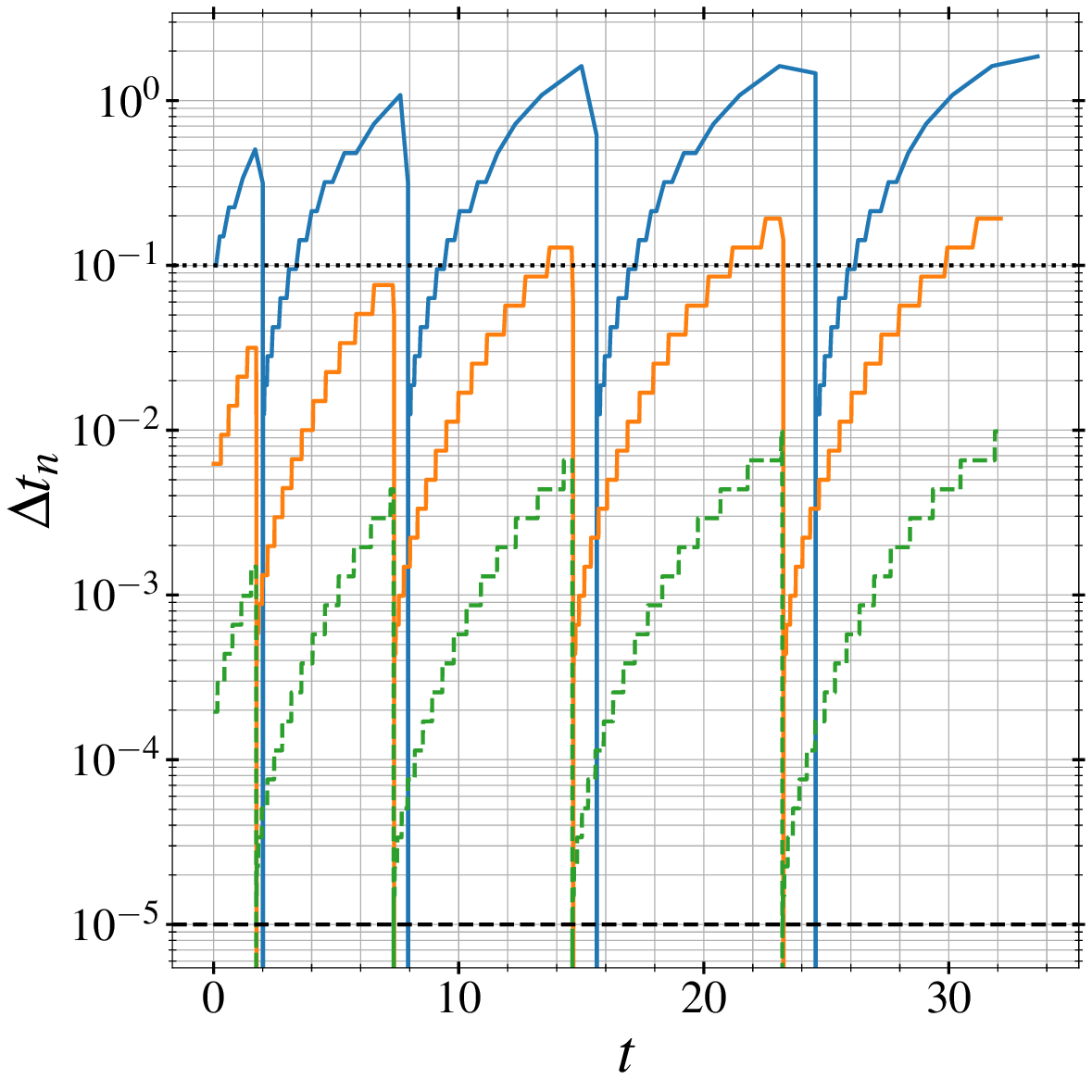}

\caption{Time step $\Delta t_n$.}
\end{subfigure}

\caption{Evolution of the LIF model using an adaptive step size algorithm
with $\Delta t_0 = 10^{-1}$ (dotted black) and $\Delta t_{\text{min}} = 10^{-5}$
(dashed black). The bounds on the pointwise error are shown:
$\chi_{\text{max}} = 2^2$ (full), $\chi_{\text{max}} = 2^{-2}$ (dashed) and
$\chi_{\text{max}} = 2^{-6}$ (dotted).}
\label{fig:results:exp2}
\end{figure}

The results can be seen in \Cref{fig:results:exp2}. At least on $[0, \tau_1]$,
we know that the solution to the LIF model with constant current $I$ is given
by the Mittag--Leffler function, which in the limit of $\alpha \to 1$ becomes
the exponential. Due to this smoothness of the solution, the time
step is allowed to increase in a stepwise fashion until a new spike is near,
when the linear interpolation results in a drastic decrease in the next step.
As expected, the bounds $(\chi_{\text{min}}, \chi_{\text{max}})$ have a
noticeable effect on the time step and allow for more accurate results. The
convergence can be easily seen empirically in \Cref{fig:results:exp2}a.

\subsection{Convergence of the \texorpdfstring{\FrAdEx{}}{FrAdEx} model}
\label{ssc:results:exp3}

Having validated the behaviour of the numerical methods from \Cref{sc:methods}
on simpler models, we can now continue to examine the convergence and performance
of the more realistic \FrAdEx{} model, for which we only consider the case
$\alpha \equiv \alpha_1 = \alpha_2$. To verify the convergence of the model,
we cannot compare to a known analytical solution, since no such solutions
exist. However, we can compare to a significantly finer discretization of the
equations, which is known as self-convergence.

As in \Cref{ssc:results:exp1}, we will look at the convergence of the spike
times and only perform the convergence study for one set of parameters
with $\alpha = 0.9$. We take the fractional capacitance
$C = \qty{100}{\pico\farad \milli\second^{\alpha - 1}}$, the constant
current $I = \qty{160}{\pico\ampere}$, the total leak conductance
$g_L = \qty{3.0}{\nano\siemens}$, the effective resting potential
$E_L = \qty{-50}{\milli\volt}$, the effective threshold potential
$V_T = \qty{-50}{\milli\volt}$, the threshold slope factor $\Delta_T = \qty{2}{\milli\volt}$,
the time constant $\tau_w = \qty{150}{\milli\second}$, the conductance
$a = \qty{4}{\nano\siemens}$, the membrane threshold potential
$V_{\text{peak}} = \qty{0}{\milli\volt}$, the reset potential
$V_r = \qty{-48}{\milli\volt}$, and the spike triggered adaptation
$b = \qty{120}{\pico\ampere}$. The systems is non-dimensionalized according to
\Cref{sc:dim}. The equation is evolved to $T = 50$ (non-dimensional)
with an initial time step $\Delta t_0 = 10^{-2}$, and adaptive parameters
$(\theta, \sigma, \rho) = (1, 0.5, 1.5)$. The adaptive error bounds are taken as
\[
\chi_{\text{min}} = \left\{\frac{1}{2^k} \middle| k \in \{0, \dots, 7\}\right\}
\quad \text{and} \quad
\chi_{\text{max}} = \left\{\frac{2}{2^k} \middle| k \in \{0, \dots, 7\}\right\}.
\]

\begin{figure}[ht!]
\begin{subfigure}{0.45\linewidth}
\centering
\includegraphics[width=0.9\textwidth]{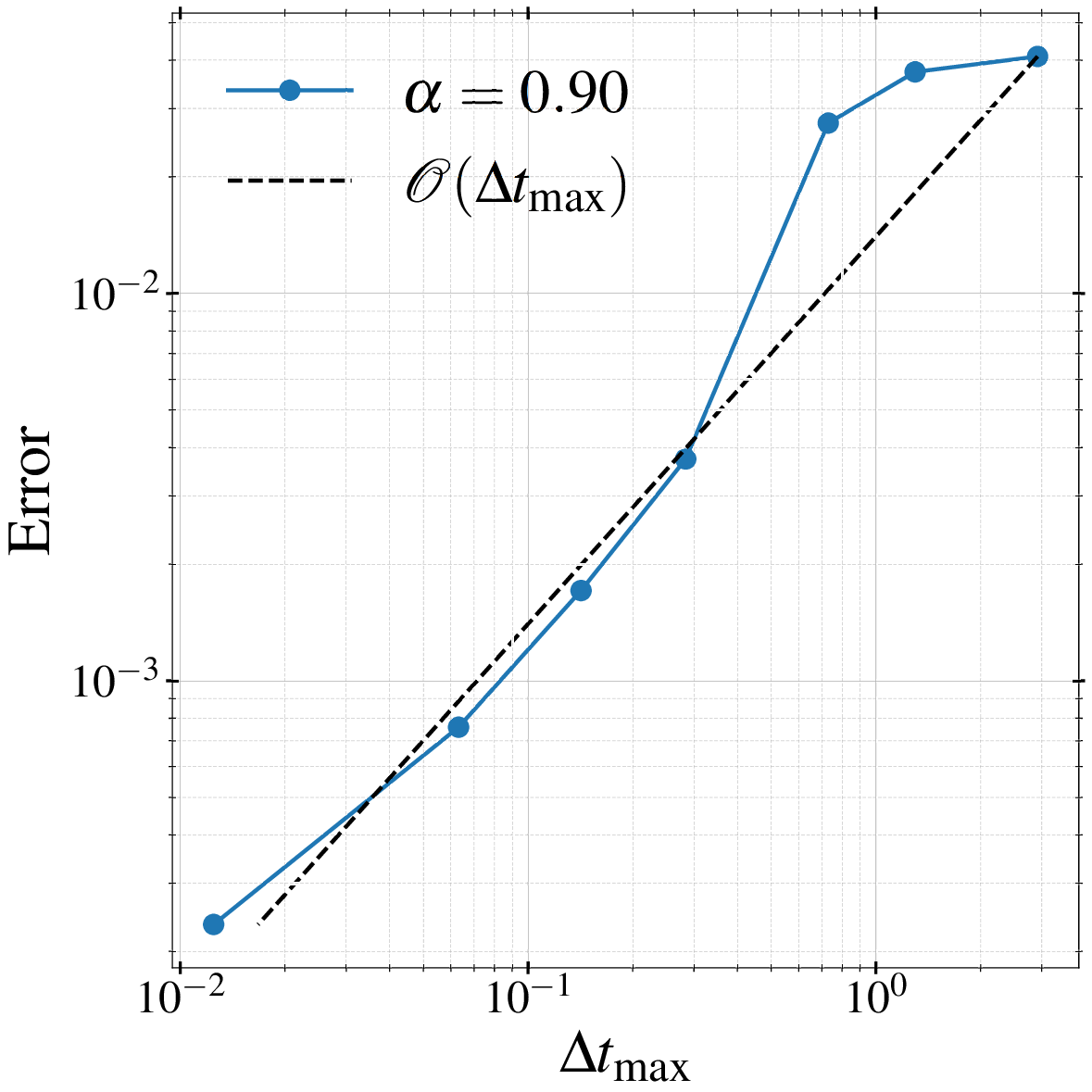}

\caption{Relative error in spike positions.}
\end{subfigure}
\begin{subfigure}{0.45\linewidth}
\centering
\includegraphics[width=0.9\textwidth]{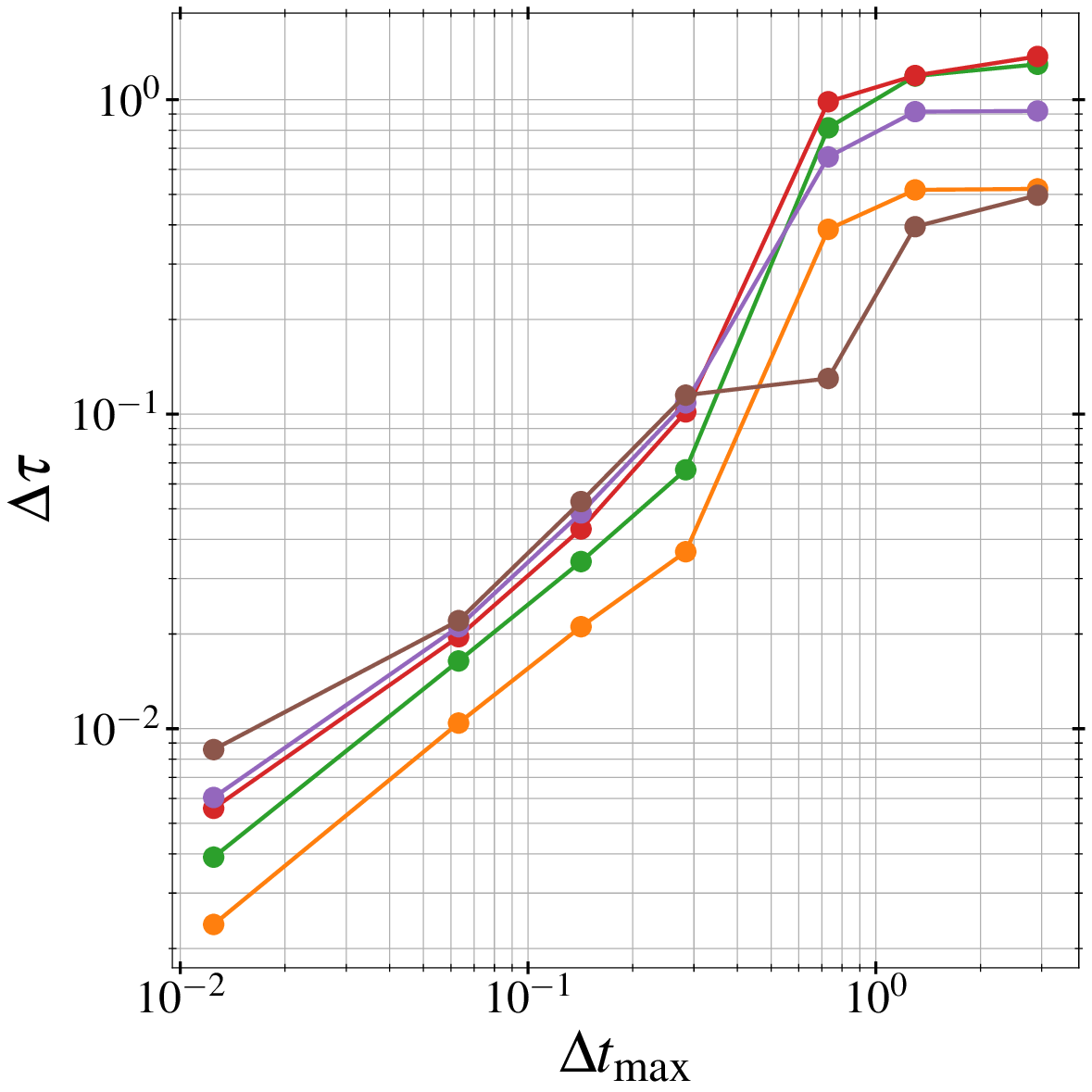}

\caption{Spike time error.}
\end{subfigure}

\caption{(a) First-order self-convergence of the \FrAdEx{} model and (b)
Pointwise error of each of the $5$ numerical spike times for every
$(\chi_{\text{min}}, \chi_{\text{max}})$ pair.}
\label{fig:results:exp3}
\end{figure}

The convergence behaviour can be seen in \Cref{fig:results:exp3}. As in the
case of the simpler PIF model, we recover first-order convergence when comparing
the finest solution with the coarser variants in our experiment. We can see in
\Cref{fig:results:exp3}b that the self-convergence error in each of the $5$
spike times follows largely the same pattern and the errors for later spike times
do not exhibit degraded convergence.

\subsection{Performance on the \texorpdfstring{\FrAdEx{}}{FrAdEx} model}
\label{ssc:results:exp4}

Finally, we briefly investigate the efficiency of the L1 method that was implemented.
We take the parameters from \Cref{ssc:results:exp3} and perform some standard
performance and scaling studies for $\alpha = 0.9$. We expect $\alpha > 0.9$
in practical applications, since lower values of the fractional order do not
present some irregular spiking behaviours, as shown in \Cref{sc:behaviour}.

For the adaptive time-stepping, we take
\[
\chi_{\text{min}} = \left\{\frac{1}{2^k} \middle| k \in \{0, \dots, 7\}\right\}
\quad \text{and} \quad
\chi_{\text{max}} = \left\{\frac{2}{2^k} \middle| k \in \{0, \dots, 7\}\right\},
\]
with $\theta = 1.0, \sigma = 0.5, \rho = 2.0$ and a minimum time step of
$\Delta t_{\text{min}} = 10^{-5}$. The equation is evolved to $T = 50$ (non-dimensional)
with an initial time step of $\Delta t_0 = 10^{-2}$. For comparison, we include
simulations with a fixed step size and compare their accuracy. The fixed time
step is taken as
\[
\Delta t_{\text{fixed}} = \frac{T}{2 N_{\text{adaptive}}},
\]
where $N_{\text{adaptive}}$ is the number of time steps taken by the adaptive
algorithm. Note that even for the fixed step size evolution, we continue to
restrict the time step close to a spike time to ensure that the Lambert W
function remains real, as discussed in \Cref{ssc:methods:spikes}. The error
for both methods is computed by means of self-convergence as described
in \Cref{ssc:results:exp3}.

\begin{figure}[ht!]
\begin{subfigure}{0.45\linewidth}
\centering
\includegraphics[width=0.9\textwidth]{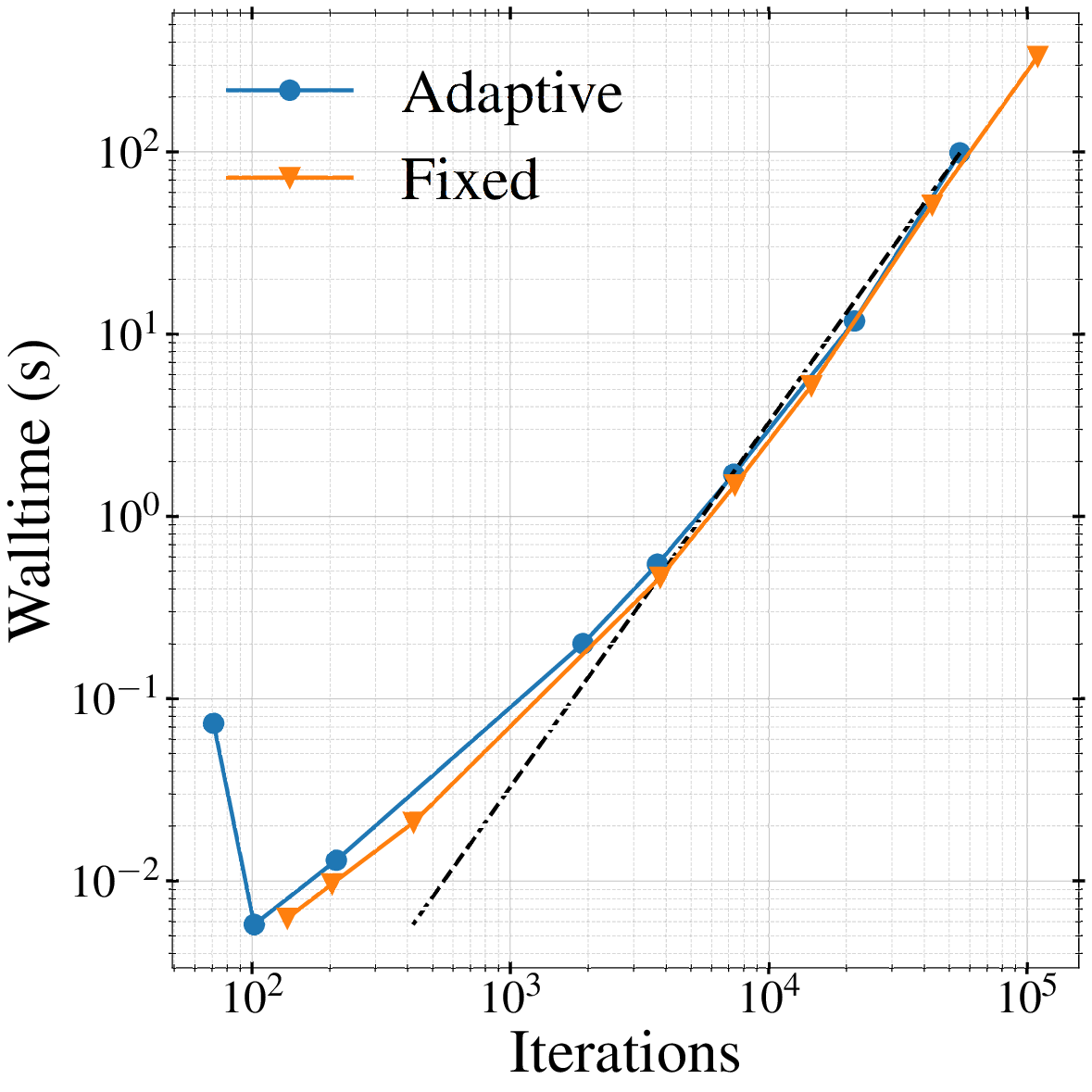}

\caption{Scaling.}
\end{subfigure}
\begin{subfigure}{0.45\linewidth}
\centering
\includegraphics[width=0.9\textwidth]{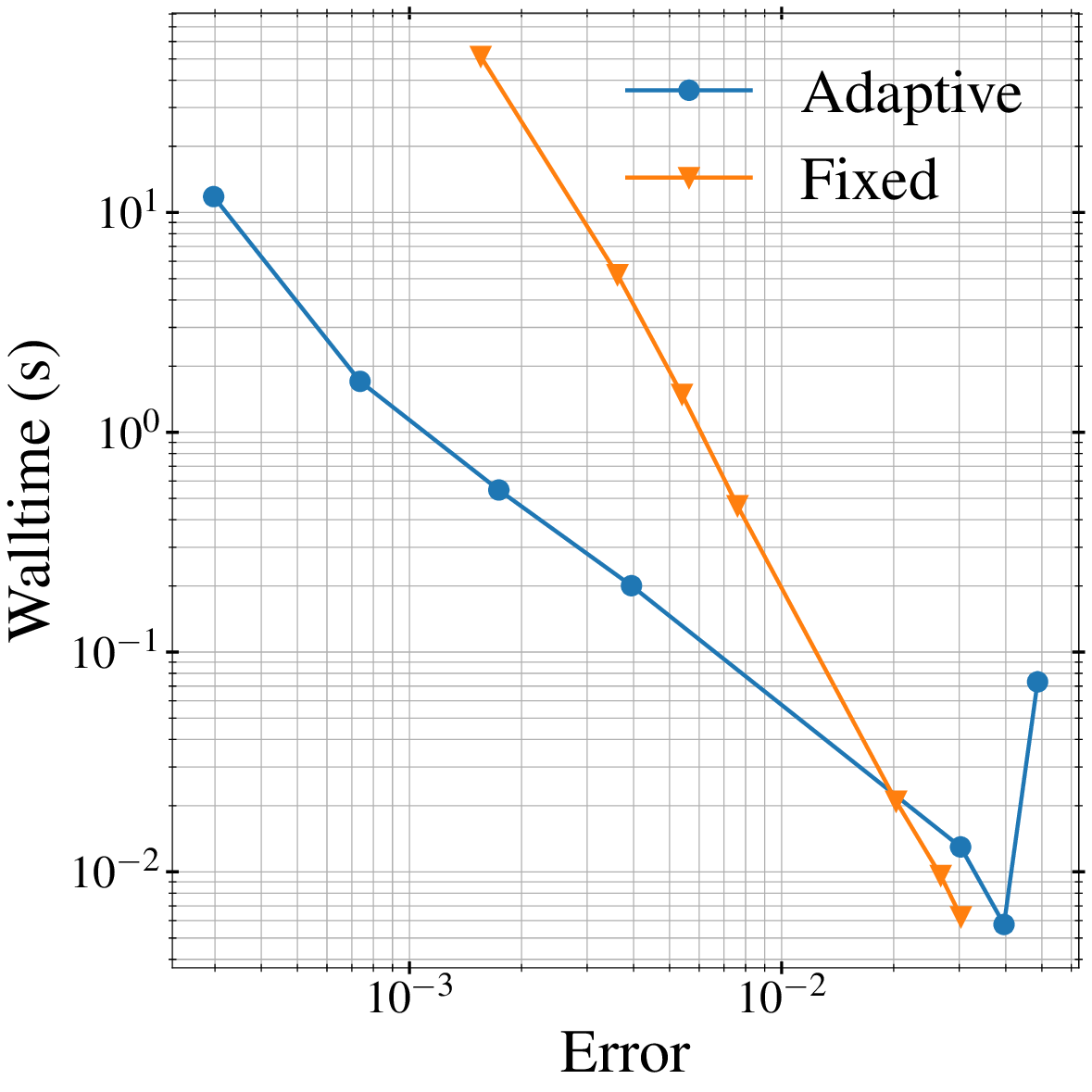}

\caption{Efficiency.}
\end{subfigure}

\caption{(a) Scaling of the L1 method with adaptive and fixed step sizes and
(b) Efficiency of the method with adaptive and fixed step sizes, which is
measured in the time required to reach a certain convergence error.}
\label{fig:results:exp4}
\end{figure}

The results of the scaling study can be seen in \Cref{fig:results:exp4}. First,
in \Cref{fig:results:exp4}a we can see that both methods scale with the
expected $\mathcal{O}(N^2)$ asymptotic order. This is due to the memory terms
inherent in fractional-order equations. For smooth solutions, the scaling can
be improved to $\mathcal{O}(N \log N)$ with the use of Fourier transforms, but
it is unclear if such methods could be used for non-uniform discontinuous systems.
Then, in \Cref{fig:results:exp4}b we can see that the fixed time step method
generally performs significantly worse than the adaptive method. To obtain
modestly small errors of $10^{-3}$, it requires about an order of magnitude more
time than the adaptive method. This is not surprising, since the fixed step
method is unable to accurately capture the exponential spiking of the model.

\section{Neural behaviour under the \texorpdfstring{\FrAdEx{}}{FrAdEx} model}
\label{sc:behaviour}

Finally, we look at simulating the fractional order \FrAdEx{} model
\eqref{eq:dim:model} with a focus on its phenomenological responses. We expect
the model to generate multiple firing patterns depending on the choice of
parameter values for different fractional exponents $\alpha$. In order to study
the range of firing responses accessible with the model, we adjust the
fractional order $\alpha$.

\begin{figure}[ht!]
\centering
\begin{subfigure}{0.3\linewidth}
\centering
\includegraphics[width=\linewidth]{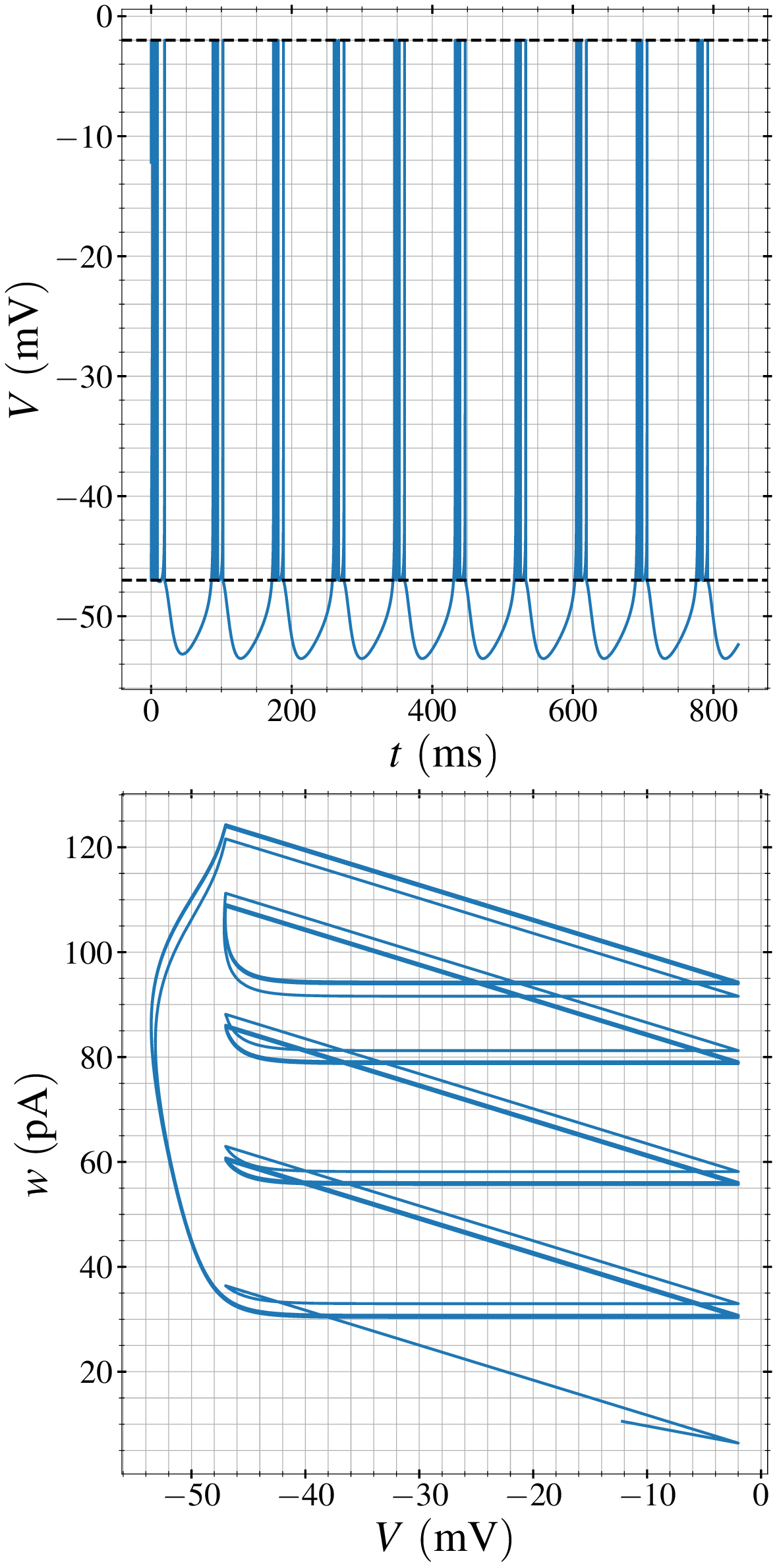}

\caption{$\alpha = 0.999$.}
\end{subfigure}
\begin{subfigure}{0.3\linewidth}
\centering
\includegraphics[width=\linewidth]{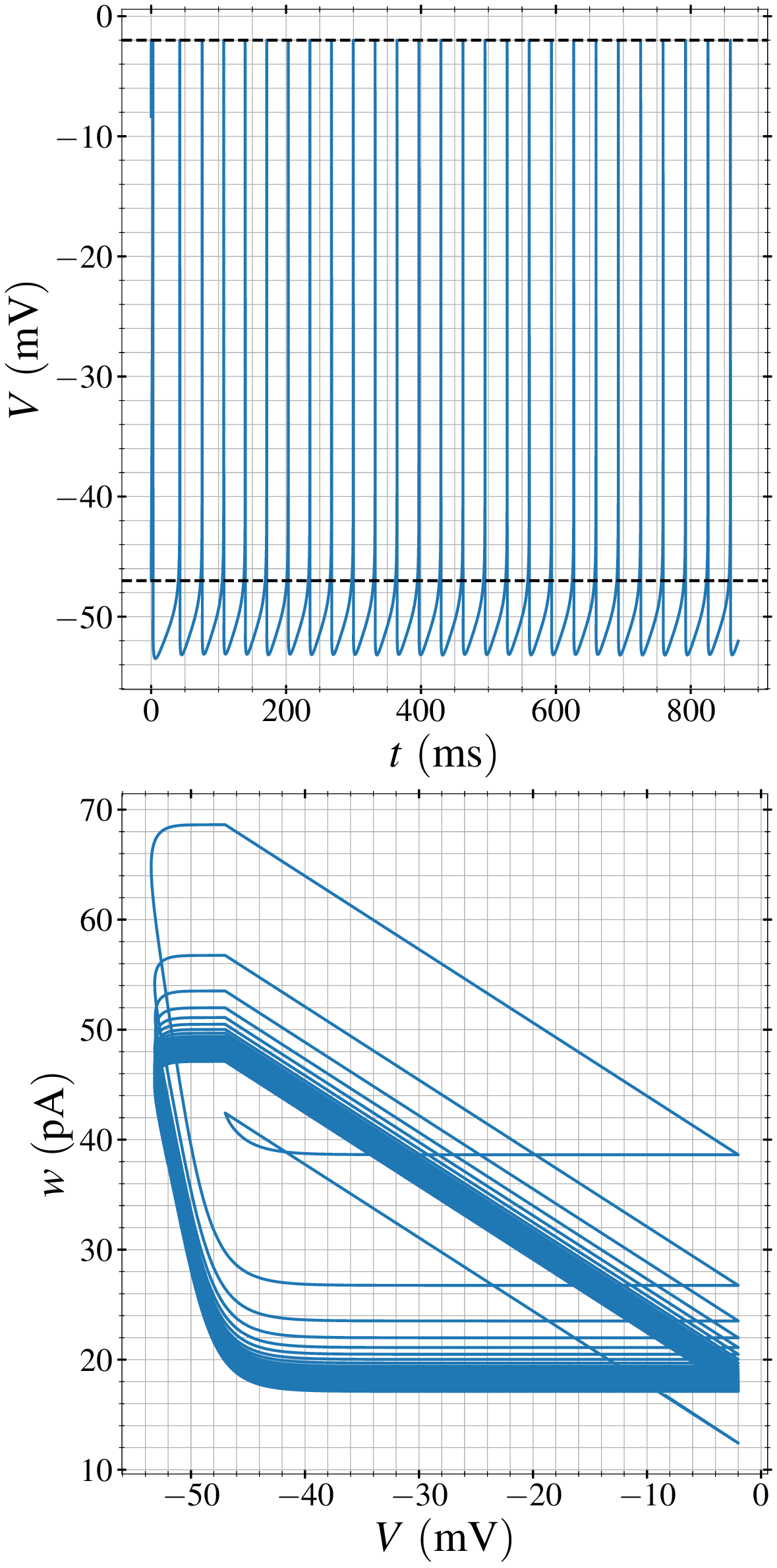}

\caption{$\alpha = 0.98$.}
\end{subfigure}
\begin{subfigure}{0.3\linewidth}
\centering
\includegraphics[width=\linewidth]{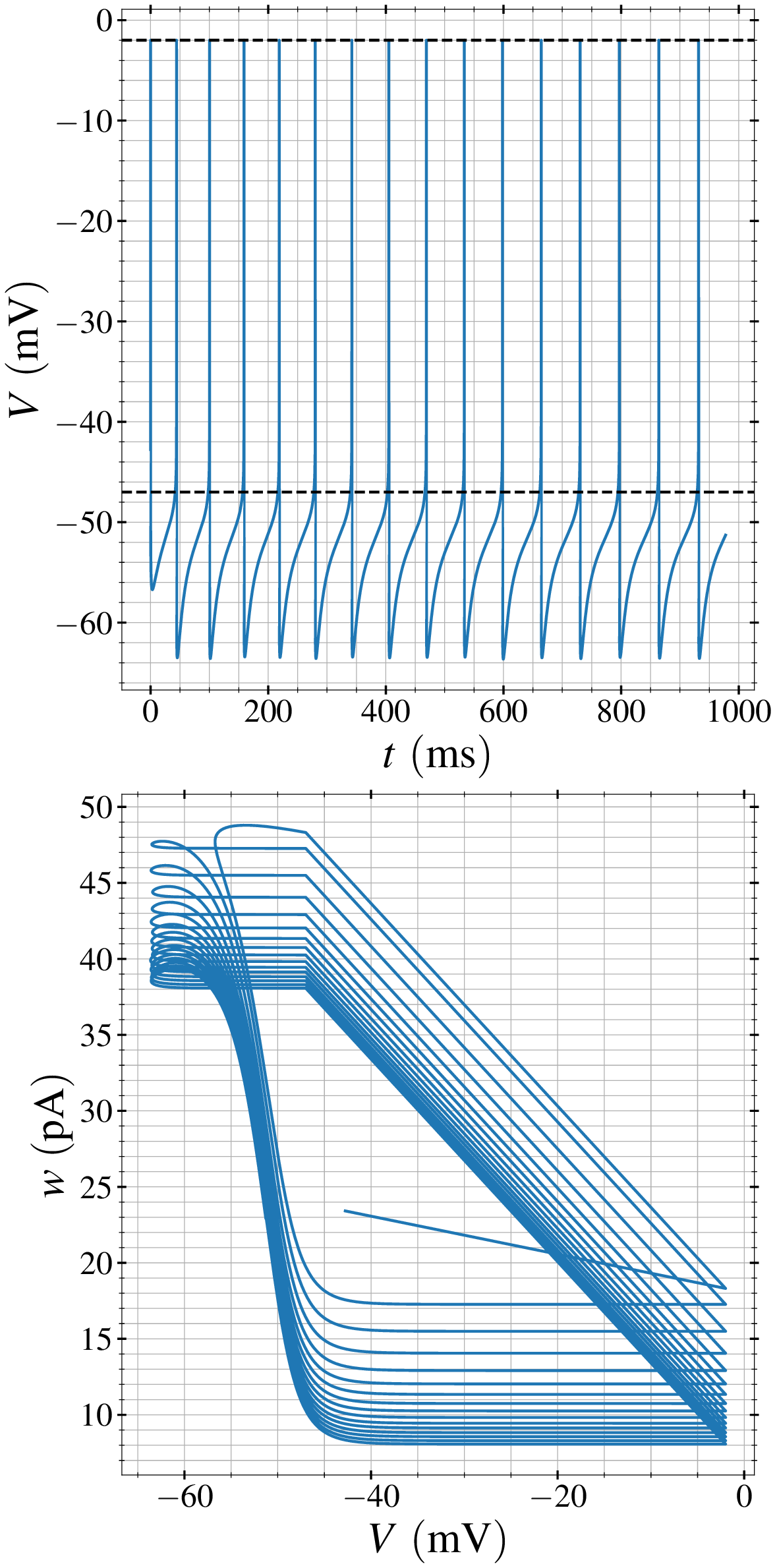}

\caption{$\alpha = 0.93$.}
\end{subfigure}

\caption{We have the evolution on $[0, 100]$ (non-dimensional) for different values
of $\alpha$ of the membrane potential (top) and the phase plane diagram (bottom).
The values of the parameters have been considered as follows $C =
\qty{100}{\pico\farad \milli\second^{\alpha - 1}}, I = \qty{160}{\pico\ampere},
g_L = \qty{12}{\nano\siemens}, E_L = \qty{-60}{\milli\volt}, V_T = \qty{-50}{\milli\volt},
\Delta_T = \qty{2}{\milli\volt}, \tau_w = \qty{130}{\milli\second},
a = \qty{-11}{\nano\siemens}, b = \qty{30}{\pico\ampere}, V_r = \qty{-48}{\milli\volt},
V_{\text{peak}} = \qty{-2}{\milli\volt}$ in the model \eqref{eq:model:fde}.}
\label{fig:behaviour:set1}
\end{figure}

In \Cref{fig:behaviour:set1}a, we show an example of a typical firing pattern
known as \emph{chattering} that is generated by setting the fractional order
close to 1, e.g. $\alpha = 0.999$. The neuron exhibits stereotypical bursting of
closely spaced spikes. As we decrease the fractional order to $\alpha = 0.98$,
the pattern switches to a fast spiking with a broad spike after-potential (SAP)
that is characterised by a small curvature after the spike (see
\Cref{fig:behaviour:set1}b). The fractional order \FrAdEx{} model can generate
adapting and tonic traces of different types \cite{connors1990intrinsic,naud2008firing}.
As we decrease $\alpha = 0.93$, we can observe fast spiking without adaptation
and with tonic spiking with sharp SAP. Here, the membrane potential increases monotonically after the rapid downswing of an action potential, as shown in
\Cref{fig:behaviour:set1}c. The corresponding phase diagrams also highlight
these changes in behaviour, as seen at the bottom of \Cref{fig:behaviour:set1}.

\begin{figure}[ht!]
\centering
\begin{subfigure}{0.3\linewidth}
\centering
\includegraphics[width=0.9\textwidth]{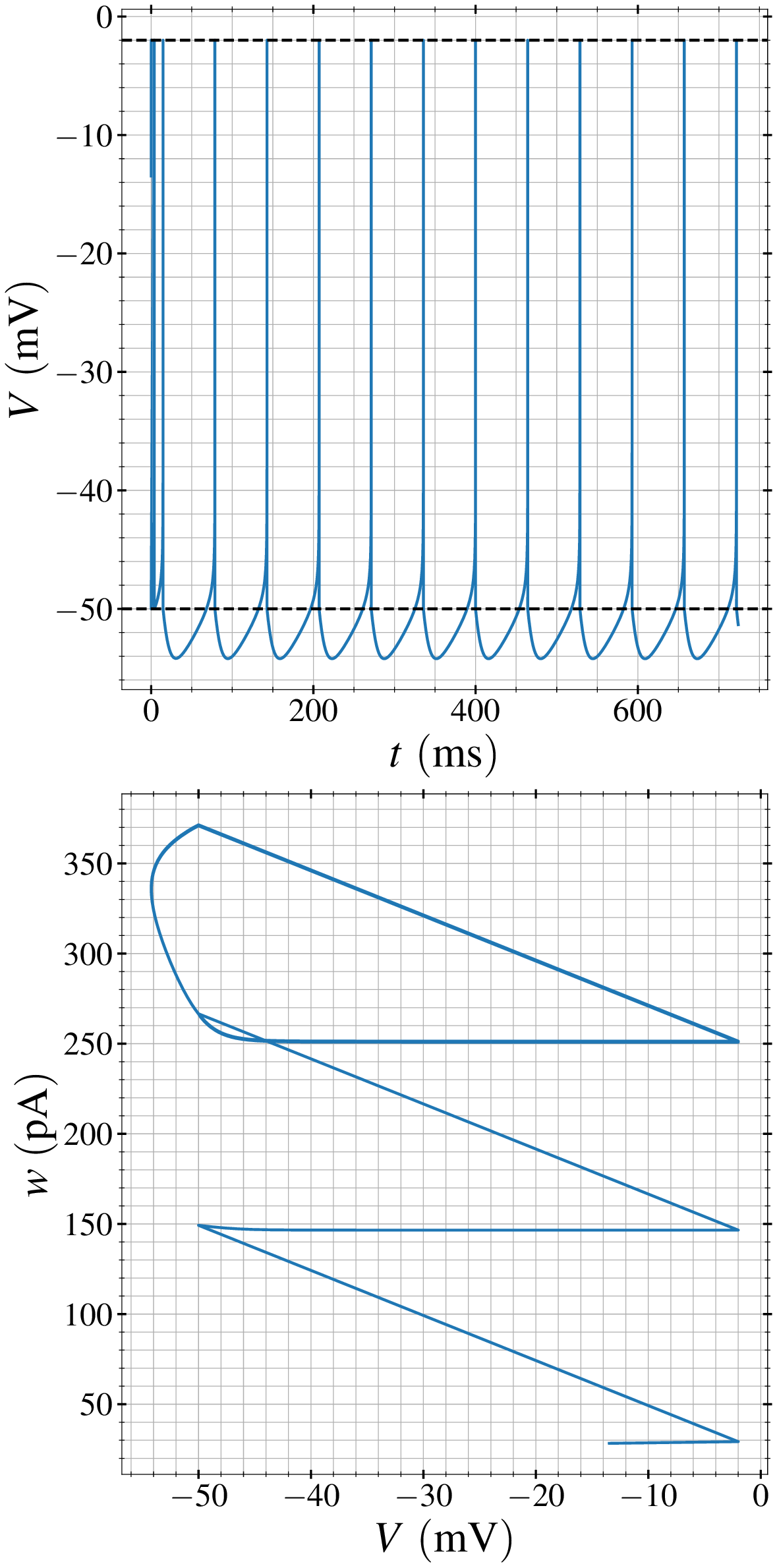}

\caption{$\alpha = 0.999$.}
\end{subfigure}
\begin{subfigure}{0.3\linewidth}
\centering
\includegraphics[width=0.9\textwidth]{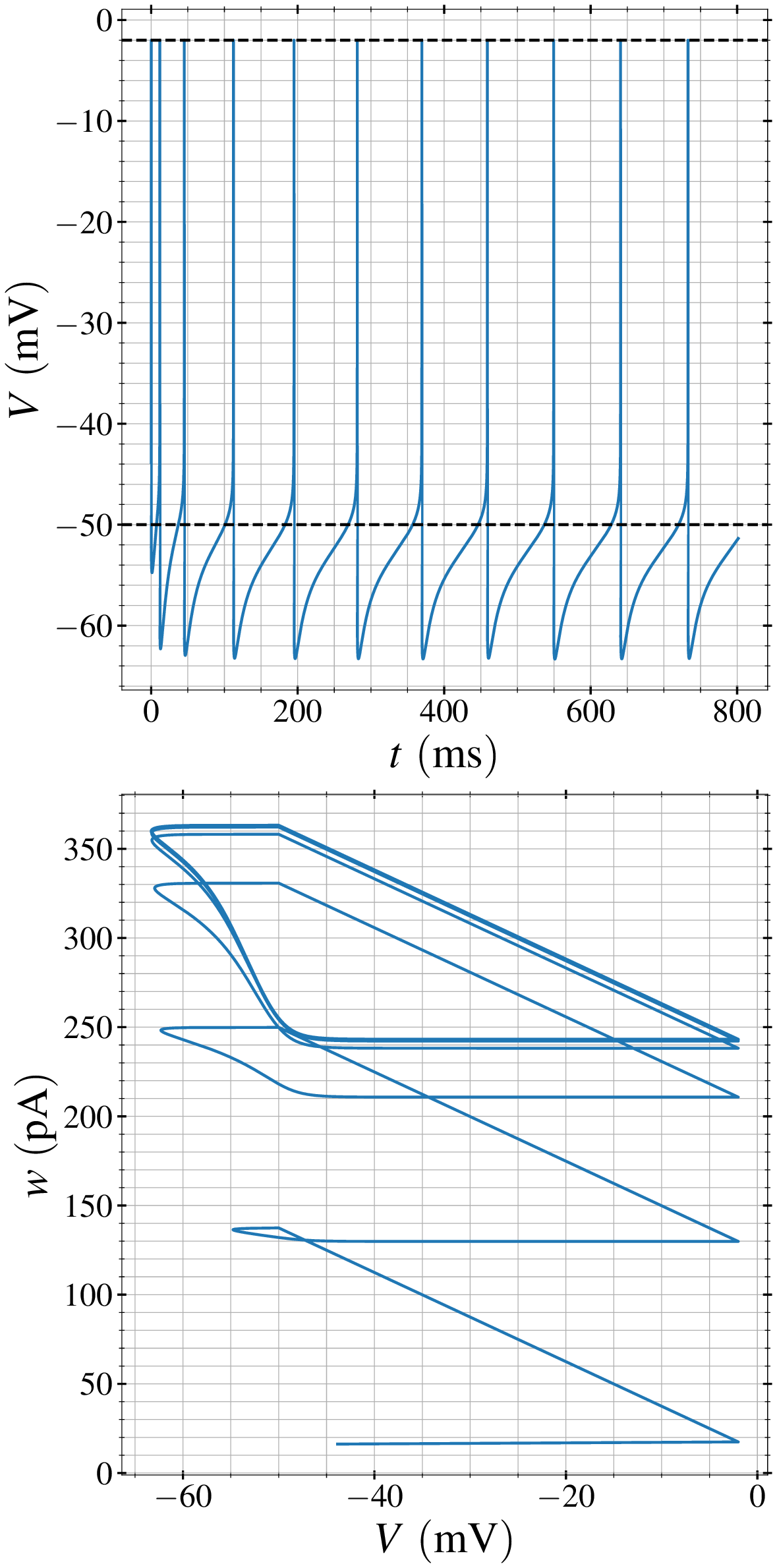}

\caption{$\alpha = 0.95$.}
\end{subfigure}
\begin{subfigure}{0.3\linewidth}
\centering
\includegraphics[width=0.9\textwidth]{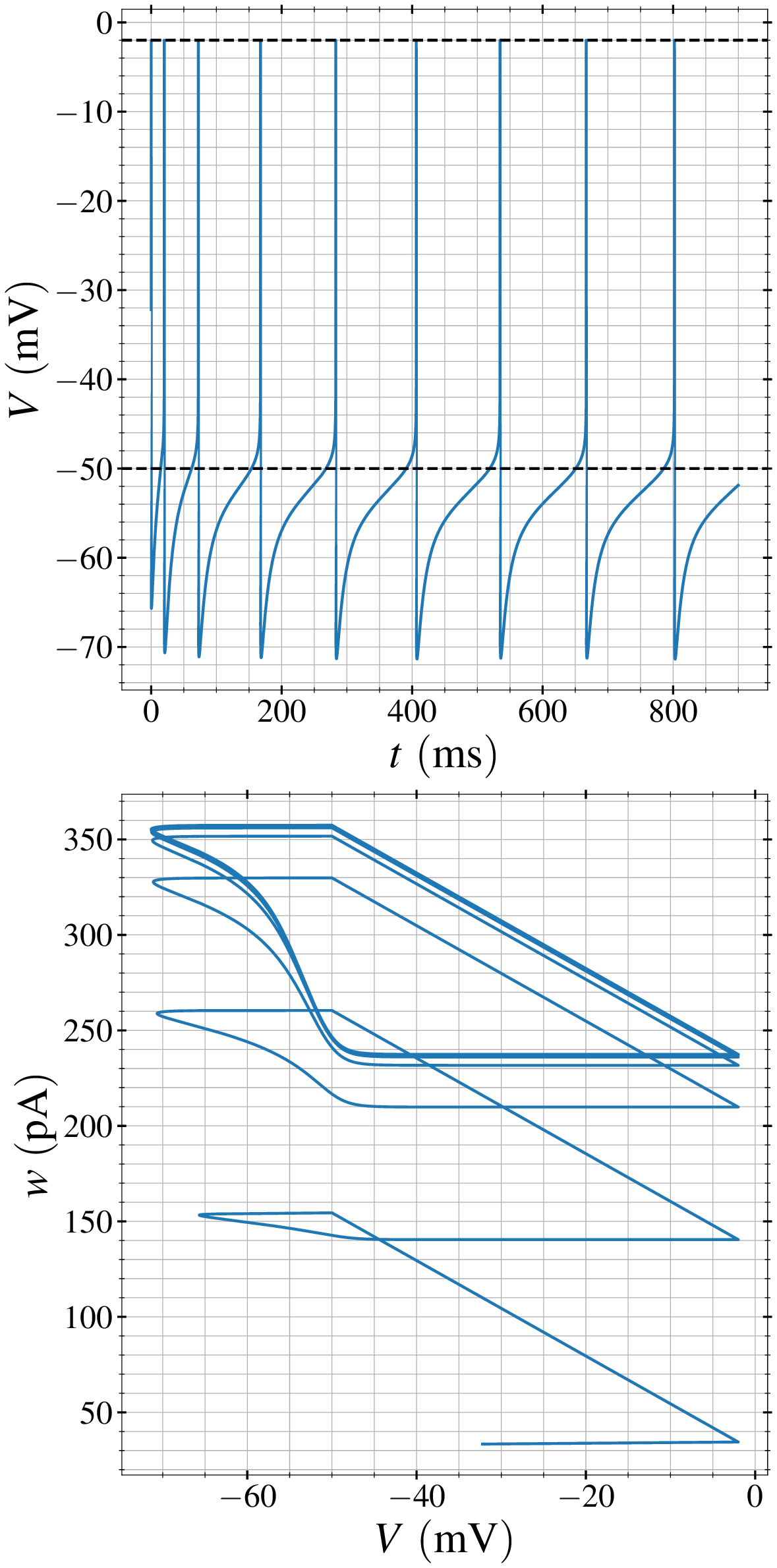}

\caption{$\alpha = 0.9$.}
\end{subfigure}

\caption{We have the evolution on $[0, 100]$ (non-dimensional) for different
values of $\alpha$ of the membrane potential (top) and the phase plane diagram
(bottom). The values of the parameters are taken to be
$C = \qty{130}{\pico\farad \milli\second^{\alpha - 1}}, I = \qty{400}{\pico\ampere},
g_L = \qty{18}{\nano\siemens}, E_L = \qty{-58}{\milli\volt}, V_T = \qty{-50}{\milli\volt},
\Delta_T = \qty{2}{\milli\volt}, \tau_w = \qty{150}{\milli\second},
a = \qty{4}{\nano\siemens}, b = \qty{120}{\pico\ampere}, V_r = \qty{-50}{\milli\volt},
V_{\text{peak}} = \qty{-2}{\milli\volt}$ in the model \eqref{eq:model:fde}.}
\label{fig:behaviour:set2}
\end{figure}

In \Cref{fig:behaviour:set2}a, we use the same fractional order $\alpha = 0.999$
with another suitable set of parameters. We observe a regular discharge of action
potential, i.e. tonic firing, and a broad spike after potential at this fractional
order. However, the behaviour switches to regular spiking with sharp SAP at
$\alpha = 0.98$ (see \Cref{fig:behaviour:set2}b) and then spike frequency
adaptation starts (i.e. the neurons fire a few spikes with short
interspike interval and then the inter-spike period increases) at $\alpha = 0.93$
(see \Cref{fig:behaviour:set2}c). The corresponding phase diagrams also highlight
these changes in behaviour, as seen at the bottom of \Cref{fig:behaviour:set2}.

\begin{figure}[ht!]
\centering
\begin{subfigure}{0.3\linewidth}
\centering
\includegraphics[width=0.9\textwidth]{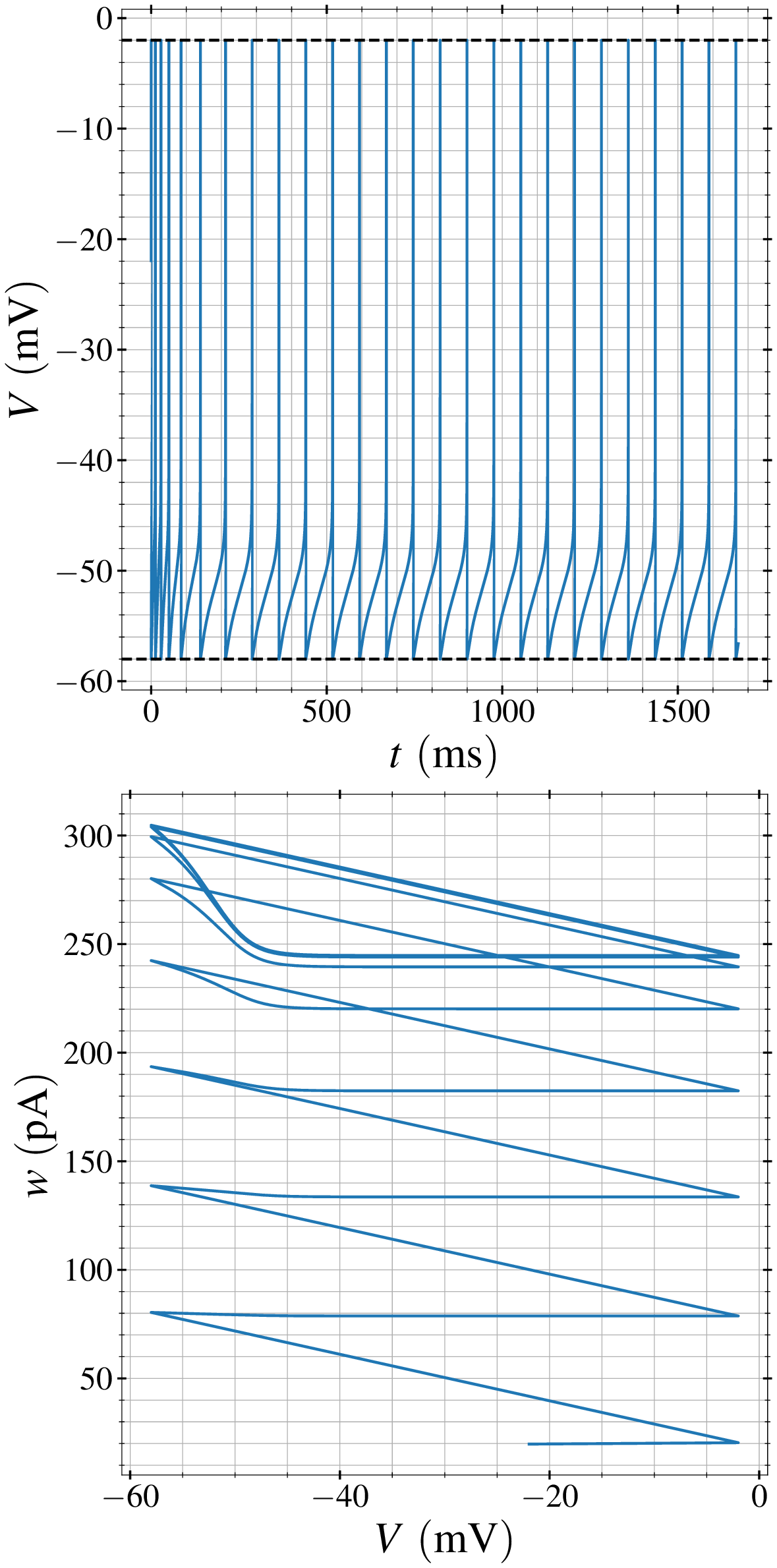}

\caption{$\alpha = 0.999$.}
\end{subfigure}
\begin{subfigure}{0.3\linewidth}
\centering
\includegraphics[width=0.9\textwidth]{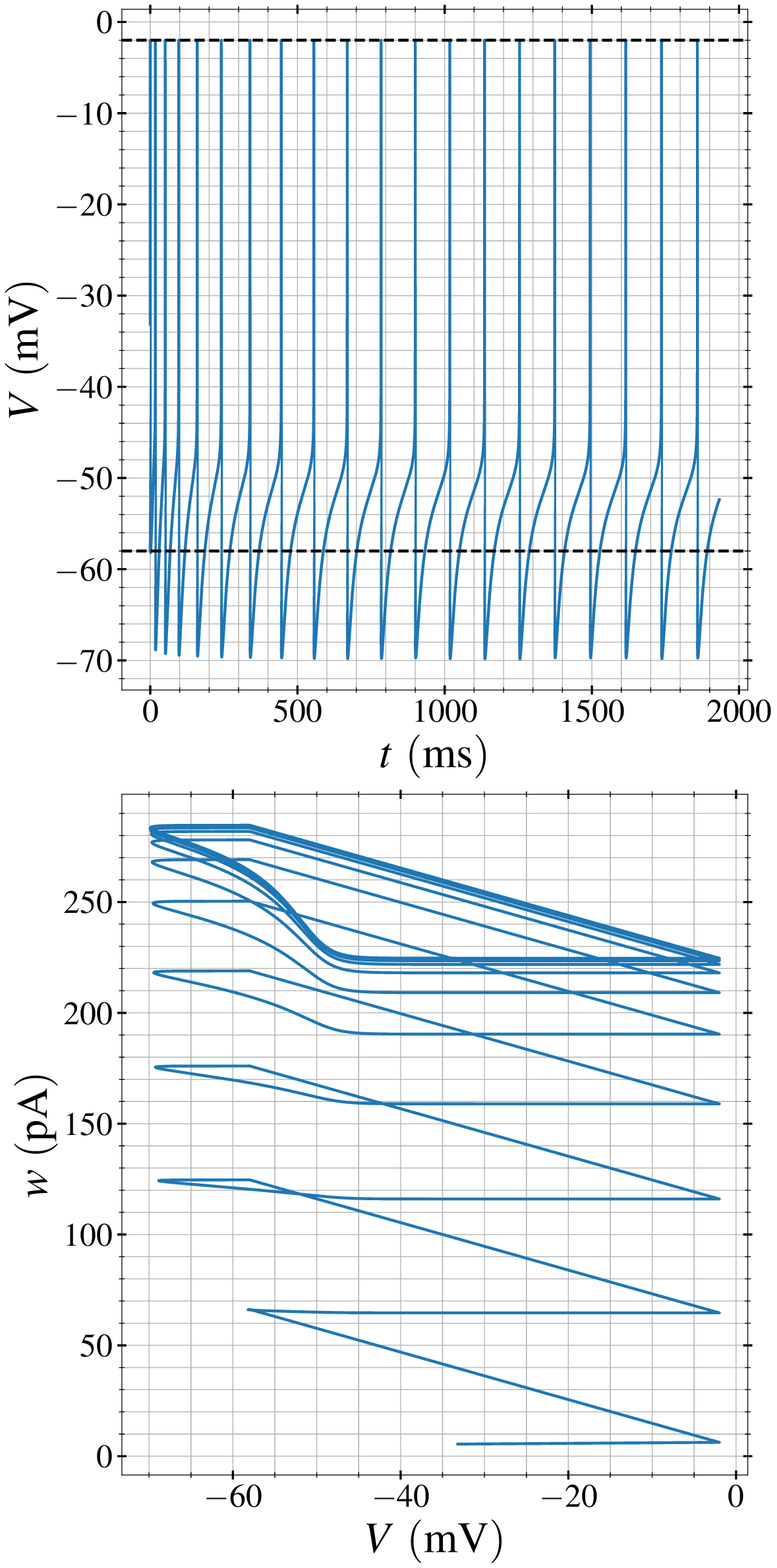}

\caption{$\alpha = 0.95$.}
\end{subfigure}
\begin{subfigure}{0.3\linewidth}
\centering
\includegraphics[width=0.9\textwidth]{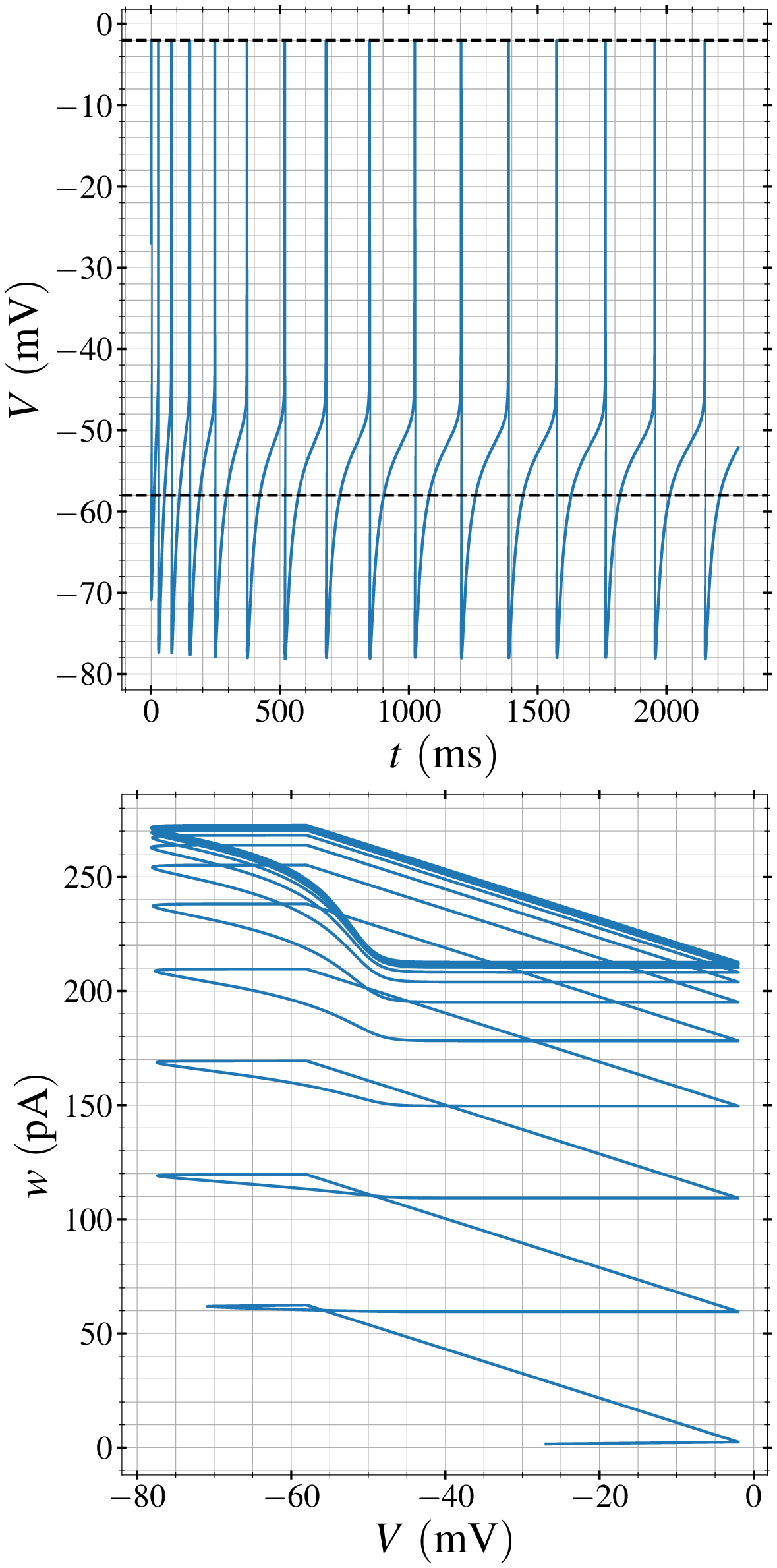}

\caption{$\alpha = 0.9$.}
\end{subfigure}

\caption{We have the evolution on $[0, 100]$ (non-dimensional) for different
values of $\alpha$ of the membrane potential (top) and the phase plane diagram
(bottom). The values of the parameters are taken to be
$C = \qty{200}{\pico\farad \milli\second^{\alpha - 1}}, I = \qty{500}{\pico\ampere},
g_L = \qty{12}{\nano\siemens}, E_L = \qty{-70}{\milli\volt}, V_T = \qty{-50}{\milli\volt},
\Delta_T = \qty{2}{\milli\volt}, \tau_w = \qty{300}{\milli\second},
a = \qty{2}{\nano\siemens}, b = \qty{60}{\pico\ampere}, V_r = \qty{-58}{\milli\volt},
V_{\text{peak}} = \qty{-2}{\milli\volt}$ in the model \eqref{eq:model:fde}.}
\label{fig:behaviour:set3}
\end{figure}

For a final example, we look at another set of parameters \Cref{fig:behaviour:set3}.
Here, the fractional order \FrAdEx{} model with $\alpha = 0.999$ can generate
intrinsic bursting, initial bursting, and then single regular spikes are emitted
(see \Cref{fig:behaviour:set3}a). As we decrease $\alpha$, we showed the transitions to tonic
spiking with sharp SAP (see \Cref{fig:behaviour:set3}b) and, finally, the model
produces regular spiking with spike frequency adaptation at $\alpha = 0.93$
(see \Cref{fig:behaviour:set3}c). The corresponding phase diagrams also highlight
these changes in behaviour, as seen at the bottom of \Cref{fig:behaviour:set3}.

\section{Conclusions}
\label{sc:conclusions}

We have discussed the introduction of a novel fractional-order adaptive exponential
(\FrAdEx{}) integrate-and-fire model for neuronal activity. This model is a
natural extension of popular adaptive exponential models that have proven
to be representative of biophysical  neurons. We have shown that the
fractional extension can also reproduce some sought after phenomenology, such as
chattering or tonic spiking of the neuron. Furthermore, we have shown that the
fractional order $\alpha$ can be used to control the behaviour of the neuronal
spiking. We know that decreasing the fractional order is equivalent to putting
more weight on the memory property of the fractional derivative. In our experiments,
this has resulted in a more regular bursting pattern, which is consistent with
similar simulations of the fractional LIF model. On the other hand, in the limit of
$\alpha \to 1$, the model will recover the integer order behaviour. Due to its
complexity and wide parameter ranges, we have not fully investigated the full
response capabilities that the \FrAdEx{}  model can achieve. This requires a more
careful theoretical analysis that is in preparation. However, we do expect that
some types of irregular spiking can only be expected in the limit of sufficiently
large $\alpha$.

The main contribution of this paper is the numerical investigation of the \FrAdEx{} model. For this, we have introduced a novel adaptive implicit L1-type numerical
method that we have fully described. As expected, at the numerical level, the
main difficulties have been the rapid exponential growth of the voltage $V$ and
the state-dependent rest conditions that the integrate-and-fire models are known for.
The exponential growth has been dealt with by making use of an implicit method,
which we have managed to solve efficiently through the use of the Lambert W
function. Previous experiments with explicit methods have not yielded satisfactory
results for this problem. The discontinuous reset condition is mainly handled
through the use of the non-uniform L1 method and adaptive time stepping. This
allows for a robust approach to the discontinuity during the exponential growth
and an accurate estimate of the spike times themselves. We have also presented a
complete error model of the method that can be directly applied to other
integrate-and-fire models. This method has been validated on several benchmark
examples that show the expected first-order convergence.

In the case of integer-order models, higher-order methods for discontinuous
ODEs are well-known. We expect to extend the method described here to second-order
for integrate-and-fire models with modest growth leading to the spike. For
example, many PIF and LIF models could be treated by using a higher-order
interpolation of the spike times and the solutions as a reset is detected. However,
the exponential models are significantly more difficult to handle, as it is
unclear how a more accurate approximation of the discontinuities can be
achieved in the fractional derivative. Another important avenue of research
pertains to considering networks of neurons described by the \FrAdEx{} model.
Depending on the coupling of the neurons, explicit solutions based on the
Lambert W function may not be possible and adaptive time stepping methods may
require additional modifications to handle the small steps leading to a
spike. Some of these issues are common to all fractional neuron models and
will no doubt see ample research in the future.

\paragraph{\textbf{Acknowledgements}}
This work was supported in part by West University of Timișoara, Romania,
START Grant No. 33580/25.05.2023 (Fikl). This work was supported in part by
CNCS-UEFISCDI, Romania, Project No. PN-IV-P1-PCE2023-0444 (Kaslik). This work
was supported in part by CSIR-HRDG, Govt. of India under
Grant No. 25/0322/23/EMR-II (Mondal).

\paragraph{\textbf{Declaration of Interests}}
The authors report no conflict of interest.


\bibliography{main}

\end{document}

%% file: main_common.tex
\usepackage[utf8]{inputenc}
\usepackage[T1]{fontenc}

\usepackage{xparse}

\usepackage{fixmath}
\usepackage{amsmath}
\usepackage{amssymb}
\usepackage{amsthm}

\usepackage{xcolor}
\usepackage{siunitx}
\usepackage{booktabs}
\usepackage{algorithm}

\usepackage[shortlabels]{enumitem}
\usepackage{booktabs}
\usepackage{subcaption}
\usepackage{algorithmic}

\usepackage{hyperref}
\usepackage[nameinlink, noabbrev]{cleveref}

\usepackage{tikz}
\usetikzlibrary{arrows.meta}
\usetikzlibrary{patterns}
\usetikzlibrary{shapes.arrows}

\crefrangeformat{theorem}{\textup{#3(#1)#4--#5(#2)#6}}

\NewDocumentCommand \newtheoremin { m m m } {
    \newtheorem{#1}{#2}
    \numberwithin{#1}{#3}
}

\newtheoremin{example}{Example}{section}
\newtheoremin{remark}{Remark}{section}
\newtheoremin{definition}{Definition}{section}
\newtheoremin{proposition}{Proposition}{section}
\newtheoremin{lemma}{Lemma}{section}
\newtheoremin{theorem}{Theorem}{section}

\crefname{algocf}{algorithm}{algorithms}
\Crefname{algocf}{Algorithm}{Algorithms}
\crefname{algocf}{algorithm}{algorithms}
\Crefname{algocf}{Algorithm}{Algorithms}

\NewDocumentCommand \FrAdEx {} {\textrm{FrAdEx}}

\NewDocumentCommand \dx { O{x} } {\,\mathrm{d} #1}
\NewDocumentCommand \vect { m } {\mathbold{#1}}
\NewDocumentCommand \od { m m } {\dfrac{\mathrm{d} #1}{\mathrm{d} #2}}

\NewDocumentCommand \CaputoD { O{\alpha} m m } {{}^{PC} D^{#1}_{#3}[#2]}
\NewDocumentCommand \RiemannD { O{\alpha} m m } {{}^{PRL} D^{#1}_{#3}[#2]}

\NewDocumentCommand \varref { m } {\overline{#1}}
\NewDocumentCommand \varhat { m } {\hat{#1}}

\NewDocumentCommand \avg { sm } {
    \IfBooleanTF#1{\langle #2 \rangle}{\left\langle #2 \right\rangle}
}
\NewDocumentCommand \jump { sm } {
    \IfBooleanTF#1{\llbracket #2 \rrbracket}{\left\llbracket #2 \right\rrbracket}
}
\NewDocumentCommand \inp { sm } {
    \IfBooleanTF#1{\langle #2 \rangle}{\left\langle #2 \right\rangle}
}

\DeclareMathOperator*{\esssup}{ess\,sup}

\ifreview
    \usepackage[mathlines]{lineno}

    \linenumbers

    \usepackage[textwidth=1.25\linewidth]{todonotes}
    \NewDocumentCommand \NoteGeneric { m m m } {\todo[
        size=\small,
        backgroundcolor=white,
        bordercolor=black,
        linecolor=#1,
        textcolor=#1]{@\textbf{#2}: #3}}

    \NewDocumentCommand \WrittenBy { m +m } {{\color{#1} #2}}
\else
    \NewDocumentCommand \NoteGeneric { m m m } {}
    \NewDocumentCommand \WrittenBy { m +m } {}
\fi
